\documentclass[conference, letterpaper]{IEEEtran}
\usepackage[utf8]{inputenc} 
\usepackage[T1]{fontenc}
\usepackage{url}
\usepackage{ifthen}
\usepackage{cite}
\usepackage[cmex10]{amsmath} 

\usepackage{caption}
\usepackage{amsmath}
\usepackage{amsthm}
\usepackage{amsfonts} 
\usepackage{float}
\usepackage{graphicx}
\usepackage{stfloats}
\usepackage{subfig}
\usepackage{float}
\usepackage{bbm}
\usepackage{amssymb}
\usepackage{xcolor}
\usepackage[short]{optidef}
\usepackage[linesnumbered,ruled]{algorithm2e}
\usepackage{bbm}
\newtheorem{definition}{Definition}
\newtheorem{theorem}{Theorem}
\newtheorem{lemma}{Lemma}
\newenvironment{proof-sketch}{\noindent{\textit {Proof Sketch:}}}{\qed}
\newtheorem{corollary}{Corollary}
\newtheorem{proposition}{Proposition}
\DeclareMathOperator*{\argmin}{arg\,min}

\title{Age-Optimal Low-Power Status Update over Time-Correlated Fading Channel}
{\author{\IEEEauthorblockN{Guidan Yao\IEEEauthorrefmark{1}, Ahmed M. Bedewy\IEEEauthorrefmark{1}, and Ness B. Shroff \IEEEauthorrefmark{1}\IEEEauthorrefmark{2}}\IEEEauthorblockA{\IEEEauthorrefmark{1}  Department of Electrical and Computer Engineering, Ohio State University} 
\IEEEauthorblockA{\IEEEauthorrefmark{2}Department of Computer Science and Engineering, Ohio State University}}}

\begin{document}
\maketitle

\begin{abstract}
In this paper, we consider transmission scheduling in a status update system, where updates are generated periodically and transmitted over a Gilbert-Elliott fading channel. The goal is to minimize the long-run average \textit{age of information} (AoI) at the destination under an average energy constraint. We consider two practical cases to obtain channel state information (CSI): (i) \emph{without channel sensing} and (ii) with \emph{delayed channel sensing}. For case (i), the channel state is revealed when an ACK/NACK is received at the transmitter following a transmission, but when no transmission occurs, the channel state is not revealed. Thus, we have to design schemes that balance tradeoffs across energy, AoI, channel exploration, and channel exploitation. The problem is formulated as a constrained partially observable Markov decision process problem (POMDP). To reduce algorithm complexity, we show that the optimal policy is a randomized mixture of no more than two stationary deterministic policies each of which is of a threshold-type in the belief on the channel. For case (ii), (delayed) CSI is available at the transmitter via channel sensing. In this case, the tradeoff is only between the AoI and energy consumption and the problem is formulated as a constrained MDP. The optimal policy is shown to have a similar structure as in case (i) but with an AoI associated threshold. Finally, the performance of the proposed structure-aware algorithms is evaluated numerically and compared with a Greedy policy.
\end{abstract}

\section{Introduction}

For status update systems, where time-sensitive status updates of certain underlying physical process are sent to a remote destination, it is important that the destination receives fresh updates.  
The \textit{age of information} (AoI) is a performance metric that is a good measure of the freshness of the data at the destination. In particular, AoI is defined as the time elapsed since the generation of the recently received status update. 

The problem of minimizing the AoI in status update systems has attracted significant recent attention (e.g., \cite{kadota2016minimizing,kadota2018scheduling,bedewy2020optimizing,qiao2019age,bedewy2020optimal,rafiee2020maintaining,bedewy2019age,bedewy2019minimizing,yao2020battle}). 
Due to the fact that sensors in the status update system are usually battery-powered and thus have limited energy supply, the problem of minimizing the long-run average AoI has to take energy constraints into account. 
Moreover, communication over a wireless channel is subject to multiple impairments such as fading, path loss and interference, which may lead to status updating failure. Since each failed transmission consumes unnecessary energy, there is a strong motivation for designing intelligent transmission scheduling algorithms i.e., retransmission or suspension of transmission to increase channel utilization as well as prolong battery life.  

Many existing works that deal with the AoI minimization problem under energy constraints in status update systems assume either perfect knowledge of the channel state or noiseless channel to guarantee successful transmission. In \cite{bacinoglu2015age, wu2017optimal}, the authors assume that the channel is noiseless, and propose offline or online status updating policies. In \cite{zhou2019joint}, the authors jointly design sampling and updating processes over a channel with perfect channel state information. The success of each transmission is guaranteed via using predefined transmission power which is a function of the channel state. However, in many practical scenarios, the channel state may not be known a priori. Thus, more recent works have also considered unreliable transmissions with imperfect knowledge of wireless channels. For example, in \cite{huang2020age}, the authors consider a block fading channel, where the channel is assumed to vary independently and identically over time slots. In \cite{ceran2019average}, the authors consider an error-prone channel, where decoding error depends only on the number of retransmissions.  


However, these works neglect an important characteristics of the wireless fading channel: The \textit{channel memory} or \textit{time correlation} \cite{tse2005fundamentals} when studying unreliable transmissions with imperfect knowledge of channel states.
Indeed, the memory can be intelligently exploited to predict the channel state and thus to design efficient scheduling policies in the presence of transmission cost. A finite state Markov chain is an often used and appropriate model for fading channel \cite{zhang1999finite}. A somewhat simplified but often-used abstraction is a two-state Markovian model known as the Gilbert-Elliot channel \cite{gilbert1960capacity}. The model assumes that the channel can be either in a good or bad state, and captures the essence of the fading process. In \cite{leng2019age}, the authors consider status updating in cognitive radio networks. The occupation of primary user's channel is modeled as a two-state Markov chain. Although a Markov chain is used to model occupation of primary channel, their threshold-type structural result is built on perfect knowledge of the channel state since update decisions are made based on perfect sensing results. In contrast, in our work, we do not assume that the channel state is known a priori at the time of making updating decisions.

Motivated by the time-correlation in a fading channel and the fact that sensors in practice are typically configured to generate status updates periodically \cite{champati2019statistical}, in this paper, we consider a status update system where the status update is generated periodically and transmitted over a Gilbert-Elliot channel. 
 We do not assume that the channel state is known a priori and consider two practical cases to obtain the channel state information (CSI): (i) \textit{(without channel sensing)} CSI is revealed by the ACK/NACK feedback of a transmission; (ii) \textit{(with delayed channel sensing)} delayed CSI is always available via delayed channel sensing regardless of transmission decisions. 
To increase the reliability of received status updates, retransmissions are allowed. With these, we study the problem of how to minimize the average AoI under a long-run average energy constraint. The problem in case (i) is formulated as a constrained partially observable Markov decision process problem (POMDP) while in case (ii), it is formulated as a constrained Markov decision problem (MDP). It is known that in general POMDP is PSPACE hard to solve and MDP suffers from the curse of dimensionality. In fact, the problem in both cases involves long-run average cost with infinite state space and unbounded costs, which makes the analysis difficult. 
Our key contributions include:
\begin{itemize}
	\item For the case without channel sensing, we show that the optimal transmission scheduling policy is a randomized mixture of no more than two \emph{stationary deterministic threshold-type} policies (Theorem \ref{thre_avg} and Corollary \ref{opt_policy_sturcture}). Note that although there are some works that deal with showing optimality of threshold-type policies in POMDPs \cite{lovejoy1987some, albright1979structural,laourine2010betting,liu2010indexability,abad2017channel}, the techniques in these papers cannot be applied to our problem. This is because, given hidden state and action, the one-stage cost in these papers is constant and bounded, while the one-stage cost in our paper depends on varying and unbounded AoI.
	\item We propose a finite-state approximation for our infinite-state (unbounded AoI and belief on channel state) belief MDP and show that the optimal policy for the approximated belief MDP converges to the original one (Theorem \ref{approx}). Based on this, we propose an optimal efficient structure-aware transmission scheduling algorithm (Algorithm \ref{alg1}) for the approximate belief MDP.
	\item For the case with delayed channel sensing, we show that the optimal transmission scheduling policy is also a randomized mixture of no more than two stationary deterministic threshold-type policies. However, due to the simplification in the state, the threshold here is on AoI (Theorem \ref{thre_avg2}). Moreover, we provide a relation between the thresholds associated with different channel states (Theorem \ref{thre_avg2}). Based on the theoretical insights, we develop an efficient structure-aware algorithm (Algorithm \ref{alg2}).
	\end{itemize}

The remainder of this paper is organized as follows. The system model is introduced in Section \ref{system_model}. For the case without channel sensing, we formulate the problem in Section \ref{formulation}, and in Section \ref{structure}, we explore the structure of the optimal policy and propose a structure-aware algorithm. In Section \ref{case22}, we investigate the case with delayed channel sensing. Section \ref{numerical_results} contains numerical results.

\section{System Model}
\label{system_model}
We consider a status update system where status updates are generated periodically and transmitted to a remote destination over a time-correlated fading channel as shown in Fig. \ref{model}. We consider a time-slotted system, where a time slot corresponds to the time duration of the packet transmission time and feedback period. Every $K$ consecutive time slots form \textit{a frame}. Updates are generated at the beginning of each frame. In any frame, if the generated status update is not delivered by the end of the frame, then it gets replaced by a new one in the next frame. Define $\mathcal{K}$ as the set of \emph{relative slot index} within a frame, $\mathcal{K}\triangleq \{1,2,\cdots, K\}$. Use $t\in \{1,2,\cdots\}$ as an \emph{absolute index} for the time slot count, which increments indefinitely with time. For any time slot $t$, the corresponding frame index $l_t\in\{1,2,\cdots\}$ is determined by $l_t=\lfloor \frac{t}{K}\rfloor+1$  and relative slot index $k_t\in\mathcal{K}$ is determined by $k_t=\left(\left(t-1\right)\!\!\!\mod K\right)+1$,
where $\lfloor \cdot \rfloor$ is the floor function.
\begin{figure}
	\centering
	\includegraphics[scale=0.28]{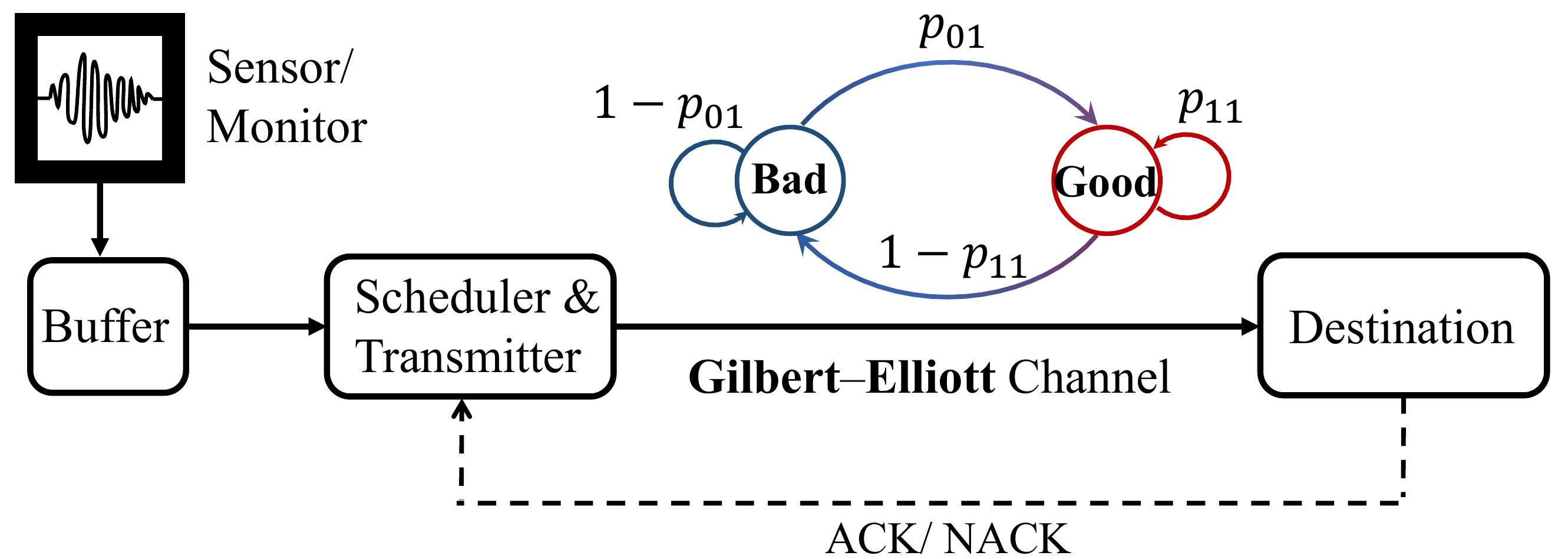}
	\caption{System Model}
	\label{model}
\end{figure}
\subsection{Channel Model} 
The time-correlated fading channel for transmission is assumed to evolve as a two-state Gilbert-Elliot model \cite{gilbert1960capacity}. Let $h_{t}$ denote the channel state at time slot $t$. Then, $h_{t}=1$ ($h_{t}=0$) denotes that channel is in a ``good'' (``bad'') state. 
In the ``bad'' state, the channel is assumed to be in a deep fade such that transmission fails with probability one; while in the ``good'' state, a transmission attempt is always successful. This assumption conforms with the signal-to-noise ratio (SNR) threshold model for reception where successful decoding of a packet at the destination occurs if and only if the SNR exceeds a certain threshold value. The channel transition probabilities are given by $\mathbb{P}(h_{t+1}\!=\!1|h_t\!=\!1)\!=\!p_{11}$ and $\mathbb{P}(h_{t+1}\!\!=\!\!1|h_t\!\!=\!0)\!=\!p_{01}$. We assume that the channel transitions occur at the end of each time slot, and that $p_{11}$ and $p_{01}$ are known. 

The presence of channel memory (time correlation) makes it possible to predict the channel state. Define Markovian channel memory as $\mu=p_{11}-p_{01}$ \cite{wang1995finite,johnston2006opportunistic}. In this paper, we assume that $p_{11}\geq p_{01}$ (positively correlated channel) (similar assumptions have been used in \cite{laourine2010betting,abad2017channel}). 

\subsection{Transmission Scheduler and Channel State Information}
\label{sys_scheduler}
At the beginning of each slot $t$, the scheduler takes a decision $u_{t} \in \mathcal{U}\triangleq\{0,1\}$, where $u_{t}=1$ means transmitting (retransmitting) the undelivered statues update, and $u_{t}=0$ denotes suspension of transmission  (retransmission).
In each frame, if the generated update is delivered at the $k_t$-th slot of the frame, then we have $u_t=0$ for the remaining slots in the frame.
For simplicity, we use transmission to refer to both transmission and retransmission in the remaining content.

In this paper, we consider two practical cases to obtain CSI: (i) \emph{(without channel sensing)} CSI is revealed via the feedback on transmission from the destination; (ii) \emph{(with delayed channel sensing)} CSI of the last time slot is always available via delayed channel sensing regardless of transmission decisions. In particular, for case (i), if a transmission is attempted, then the scheduler receives an error-free ACK/NACK feedback from the destination specifying whether the status update was delivered or not before the end of the slot. We use $\Theta$ to denote the set of observations, $\Theta\triangleq\{0, 1\}$. Let $\theta_t \in \Theta$ be the observation at time slot $t$. Then, $\theta_{t} =1$ denotes a successful transmission. $\theta_{t} =0$ occurs when the transmission occurs over the channel in the bad state or the transmission is suspended. Note that when a decision is made not to transmit updates, the scheduler will not obtain feedback revealing the CSI. Thus, the channel in this case is partially observable. In contrast, for case (ii), CSI of the last time slot is always available via delayed channel sensing regardless of transmission decisions. 
\subsection{Age of Information} 
Age of information (AoI) reflects the timeliness of the information at the destination. It is defined as the time elapsed since the generation of the most recently received update at the destination. Let $\Delta_{t}$ denote the AoI at the beginning of the time slot $t$. Let $U(t)$ denote the generation time of the last successfully received status update for time slot $t$. Then, $\Delta_{t}$ is given by $\Delta_{t}\triangleq t-U(t)$.

If a status update is not successfully delivered in slot, then the AoI increases by one, otherwise, the AoI drops to the time elapsed since the beginning of the frame (generation time of the newly delivered status update). 
Then, the value of $\Delta_{t+1}$ is updated as follows:
\begin{align}
\label{age_update}
& \Delta_{t+1}=
  \begin{cases}
  k_t & \text{if}\, \, u_{t}=1, \theta_{t}=1, \\  
  \Delta_{t}+1 & \text{otherwise}. 
  \end{cases} 
  \vspace{-0.2cm}
\end{align}
Let $\mathcal{A}_k$ denote the set of all possible AoI values at the $k$-th slot of a frame. By \eqref{age_update}, $ \mathcal{A}_k\!=\!\{\Delta:\!\Delta=mK\!+\!(k)_{-}, m\in \{0,1,2,\cdots\}\}$, where $(k)_-\!\triangleq((K+k-2)\!\mod\! K) +1$ denotes the relative slot index before $k$.
An example of the AoI evolution with K=4 is illustrated in Fig. \ref{evolution_AoI}. 
\begin{figure}
\centering
	\includegraphics[width=0.45\textwidth]{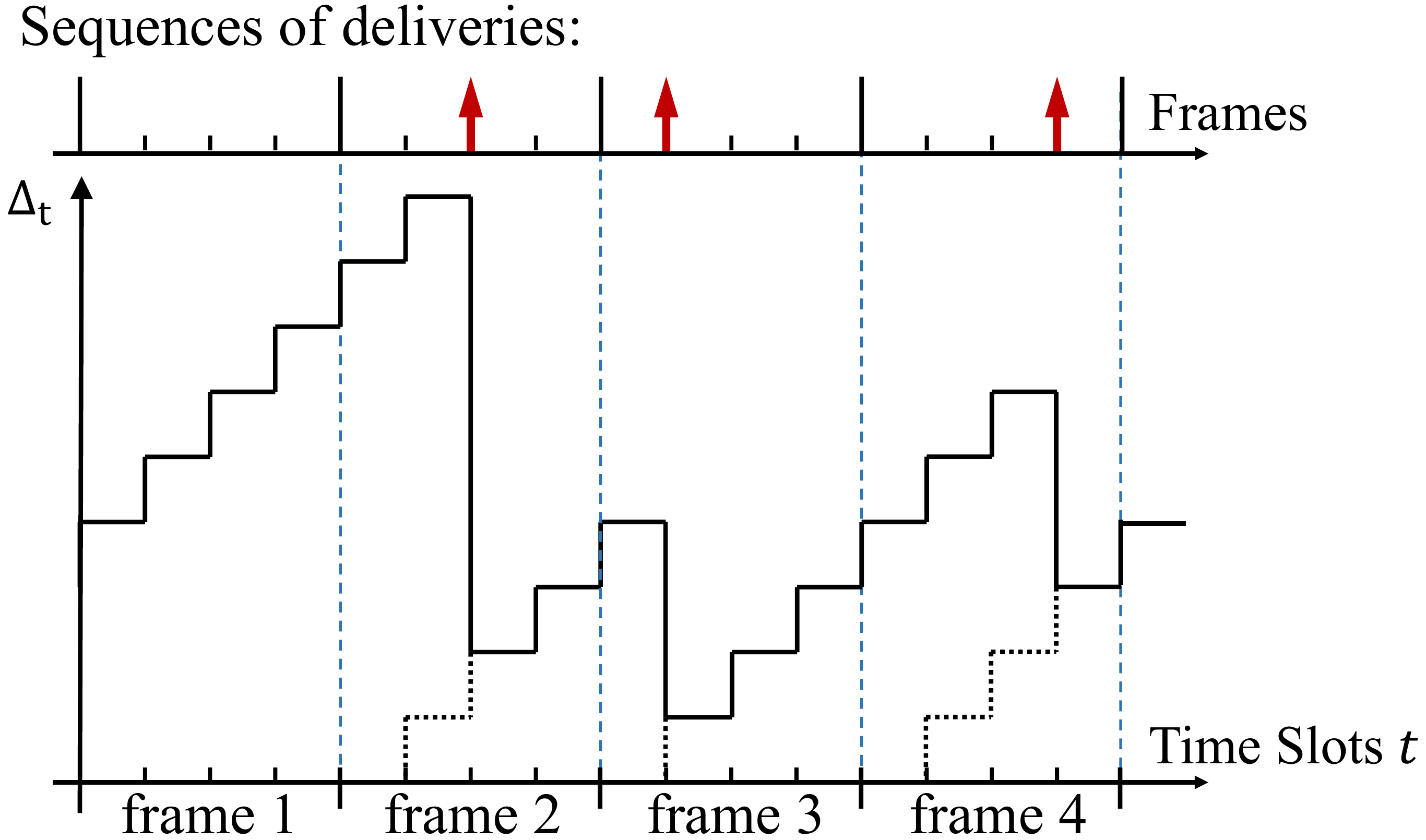}
	\caption{\small{On the top, a sample sequence of deliveries during four frames. Each frame consists of 4 time slots. The upward arrows represent the times of deliveries. On the bottom, the associated evolution of AoI.}}
\label{evolution_AoI}
\end{figure}

We aim to design an energy efficient scheduler, where each transmission consumes one unit energy. Therefore, the long-run average energy consumption cannot exceed a certain limit $E_{\text{max}}\in (0,1]$. Observe that $E_{\text{max}}=1$ means that we have enough energy to support a transmission in every time slot. Although a failed transmission does not decrease AoI, it provides channel state information at the cost of energy. Thus, the transmission scheduler has to balance tradeoffs across energy, AoI, channel exploration, and channel exploitation.

\section{Constrained POMDP Formulation and Lagrangian relaxation without Channel Sensing}
\label{formulation}
\subsection{Constrained POMDP Formulation}
At the beginning of each time slot, the scheduler chooses an action $u$. Given that the state of the underlying Markov channel is $i$, the user observes $\theta(i, u)\in \{0,1\}$, which indicates the state of the current channel. Specifically, an ACK will be received if and only if the status update is transmitted over a ``good" channel, i.e. $\theta(1,1)\!=\!1$. Otherwise, for $(i, u)\neq (1,1)$, $\theta(i, u)\!=\!0$. Upon receipt of the feedback/observation, the AoI changes accordingly at the end of this slot. The sequence of operations in each slot is illustrated in Fig. \ref{sequ_of_oper}. Note that when transmission is suspended, the channel state is not directly observable. Together with the average energy constraint, the problem we consider in the paper turns out to be a constrained partially observable Markov decision problem (POMDP). 
\begin{figure}[h]
	\centering
	\includegraphics[scale=.39]{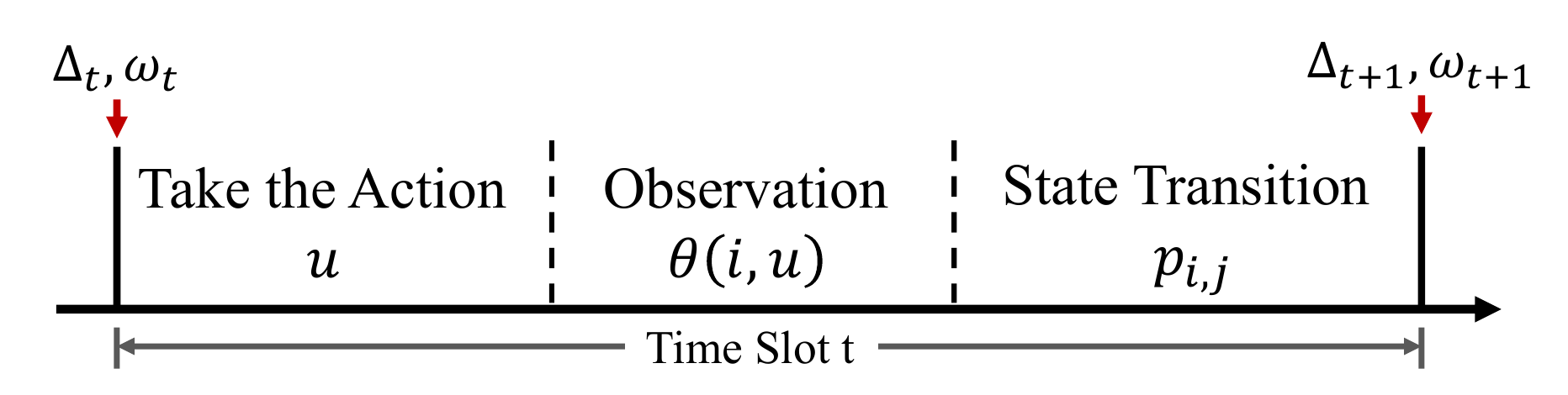}
	\caption{Sequence of operations in a slot}
	\label{sequ_of_oper}
\end{figure}

It has been shown in \cite{smallwood1973optimal} that for any slot $t$, a belief state $\omega_t$ is a \textit{sufficient statistic} to describe the knowledge of underlying channel state and thus can be used for making optimal decisions at time slot $t$. 

\begin{definition}
The belief state $\omega_t$ is the conditional probability (given observation and action history) that channel is in a good state at the beginning of the time slot $t$.
\end{definition}
\noindent Thus, adding the belief to the system state, the constrained POMDP can be written as constrained belief MDP \cite{sawaragi1970discrete}. We describe the components of the framework as follows:

 \textbf{States:} The system state consists of completely observable states and the belief state, i.e., the system state at slot $t$ is defined by a 3-tuple $\mathbf{s}_{t}\!\!=\!\! (\Delta_{t}, k_{t}, \omega_{t})$, where $\Delta_t \!\!\in\! \mathcal{A}_{k_t}$ is the AoI state that evolves as  \eqref{age_update}; $k_t\!\in\! \mathcal{K}$ is the relative slot index in the frame $l_t$ that evolves as $k_{t+1}\!=\!(k_{t})_+$, where $(y)_+\triangleq (y\mod K)+1$; $\omega_t$ is the belief state whose evolution is defined in the following paragraph. 

\textit{Belief Update:} 
Given $u_t$ and $\theta_t$, the belief state in time slot $t+1$ is updated by $\omega_{t+1}=\Lambda(\omega_t,u_t,\theta_t)$, where $\Lambda(\omega_t,u_t,\theta_t)$ is given by 
\vspace{-0.0cm}
\begin{align}
\vspace{-0.4cm}
& \omega_{t+1}=\! \Lambda(\omega_t,u_t,\theta_t)=\!
  \begin{cases}
  p_{11}\! & \text{if}\,  u_{t}=1, \theta_{ t}=1, \\  
  p_{01}\! & \text{if}\,  u_{t}=1, \theta_{ t}=0, \\
  \mathcal{T}(\omega_{t})\! & \text{if}\,  u_{t}=0,
  \end{cases} 
  \label{omega}
\end{align}
where $\mathcal{T}(\omega_{t})= \omega_{t} p_{11}+(1-\omega_{t})p_{01}$ denotes the one-step belief update. Observe that, if $u_t=0$, then the scheduler will not learn the channel state and the belief is updated only according to the Markov chain. If $u_t=1$, the observation $\theta_t$ after the transmission provides the true channel state before the state transition, which occurs at the end of the time slot (see Fig. \ref{sequ_of_oper}).

Let $\mathcal{T}^m(\omega_{t})\triangleq \mathbb{P}(h_{t+m}=1|\omega_{t})$ denote $m$-step belief update when the channel is unobserved for $m$ consecutive slots, where $m\in\{0,1,\cdots\}$ and $\mathcal{T}^0(\omega)=\omega$. 
Note that by \eqref{omega}, after a transmission ($u_t=1$), $\omega_{t+1}$ is either $p_{01}$ or $p_{11}$. The belief state $\omega$ is, hereafter, updated by $\mathcal{T}$ upon each suspension until next transmission attempt. Thus, the belief state $\omega$ is in the form of $\mathcal{T}^m(p_{01})$ or $\mathcal{T}^m(p_{11})$, where $m\geq 0$. Moreover, an increase in AoI by one results from either a failed transmission or suspension. Thus, given AoI state $\Delta_t$, the maximum suspension time after last transmission is no longer than $\Delta_t-1$. 
 By this, given AoI state $\Delta$, the belief state belongs to the following set $\Omega_\Delta\triangleq \{\omega: \omega=\mathcal{T}^m(p_{01}) \ \text{or} \ \mathcal{T}^m(p_{11}), 0\leq m<\Delta\}$. 
As a result, the state space is given by $\mathcal{S}\triangleq \{(\Delta,k,\omega):k\in K, \Delta\in \mathcal{A}_k,  \omega\in \Omega_\Delta\}$. 

 \textbf{Actions:} Action set is $\mathcal{U}=\{0,1\}$ defined in Section \ref{sys_scheduler}. 

 \textbf{Transition probabilities:} Given the current state $\mathbf{s}_{t}= (\Delta_{t}, k_{t}, \omega_{t})$ and action $u_{t}$ at slot $t$, the transition probability to the state $\mathbf{s}_{t+1}=(\Delta_{t+1}, k_{t+1}, \omega_{t+1})$ at the next slot $t+1$, which is denoted by $P_{\mathbf{s}_{t}\mathbf{s}_{t+1}}(u_t)$, is defined as
 \vspace{-0.2cm}
 \begin{align}
\vspace{-0.2cm}
	P_{\mathbf{s}_{t}\mathbf{s}_{t+1}}(u_t)&\triangleq\mathbb{P}(\mathbf{s}_{t+1}|\mathbf{s}_{t},u_{t})\notag\\
	&=\sum_{\theta_{t}\in \Theta}\mathbb{P}(\theta_{t}|\mathbf{s}_{t},u_{t})\mathbb{P}(\mathbf{s}_{t+1}|\mathbf{s}_{t},u_{t},\theta_{t}),
	\label{transition_prob}
\end{align}
\vspace{-0.1cm}
where
\vspace{-0.1cm}
\begin{align}
\vspace{-0.3cm}
 &\mathbb{P}(\theta_{t}|\mathbf{s}_{t},u_{t})=
  \begin{cases}
  \omega_{t} & \text{if}\, \, u_{t}=1,\theta_{t}=1, \\
  1-\omega_{t} & \text{if} \, \, u_{t}=1,\theta_{t}=0,\\  
  1 &\text{if} \, \, u_{t}=0,\theta_{t}=0,\\
  0 & \text{otherwise},
  \end{cases} 
\end{align}
\vspace{-0.15cm}
\vspace{-0.25cm}
\begin{align}
\vspace{-0.3cm}
&\mathbb{P}(\mathbf{s}_{t+1}|\mathbf{s}_{t},u_{t},\theta_{t})\notag\\
=&\! \begin{cases}
  1 &\!\!\! \text{if}\, \,\mathbf{s}_{t+1}\!=\!(k_{t},(k_{t})_{+},\Lambda(\omega_t,u_t,\theta_t)), u_{t}\!=\!1,\theta_{t}\!=\!1, \\
  1 & \!\!\!\text{if} \, \,\mathbf{s}_{t+1}\!=\!(\Delta_{t}\!+\!1,(k_{t})_{+},\Lambda(\omega_t,u_t,\theta_t)), \theta_{t}\!=\!0,\\  
  0 &\! \!\!\text{otherwise}.
  \end{cases} 
  \vspace{-0.1cm}	
\end{align}
\textbf{Costs:} Given a state $\mathbf{s}_t=(\Delta_{t}, k_{t}, \omega_{t})$ and an action choice $u_t$ at slot $t$, the cost of one slot is the AoI at the beginning of this slot, i.e., we have
$
	C_\Delta(\mathbf{s},u_{t})=\Delta_{t}
$.
Moreover, the energy consumption of one slot is
$
	C_E(\mathbf{s},u_{t})=u_t
$.

A transmission scheduling policy $\pi=\{d_1,d_2,\cdots\}$ specifies the decision rule for each time slot, where a decision rule $d_t$ maps the history of states and actions, and the current state to an action. A policy is \emph{stationary} if the decision rule is independent of time, i.e., $d_{t}=d$, for all $t$. Moreover, a policy is \emph{randomized} if $d_t: \mathcal{S}\rightarrow \mathcal{P}(\mathcal{U})$ specifies a probability distribution on the set of actions. The policy is \emph{deterministic} if $d_t: \mathcal{S} \rightarrow \mathcal{U}$ chooses an action with certainty. 
For any policy $\pi$, we assume that the resulted Markov chain is a unichain (same assumptions are also made in \cite{zhou2019joint,djonin2007mimo}). Our objective is to design a policy $\pi$ that minimizes the long-run average AoI $\bar{A}(\pi)$ while the long-run average energy consumption $\bar{E}(\pi)$ does not exceed $E_{\text{max}}$, which is formulated as 

\textit{Problem 1 (Constrained average-AoI belief MDP):}
\begin{alignat}{2}
\label{ori_prob}
 \bar{A}^\star\triangleq   \min_\pi \quad & \bar{A}(\pi)= \limsup_{T \rightarrow \infty} \frac{1}{T}\mathbb{E_\pi}\big[\sum_{t=1}^T C_\Delta(\mathbf{s}_{t},u_{t})\big] &  \\
    \mathrm{s.t.} \quad & \bar{E}(\pi)=\limsup_{T \rightarrow \infty} \frac{1}{T}\mathbb{E_\pi}\big[\sum_{t=1}^T C_E(\mathbf{s}_{t},u_{t})\big] \leq E_{\text{max}} & .
        \nonumber 
\end{alignat}
We use $\bar{A}^\star$ to denote the optimal average AoI, which is the solution to the problem \eqref{ori_prob}. 
We show in Section \ref{structure} that there exists a stationary policy which is a randomized mixture of no more than two deterministic policies that achieves $\bar{A}^*$.

\subsection{Lagrange Formulation of the Constrained POMDP}
\label{Lag_transform}
To obtain the optimal transmission scheduling policy, we reformulate the constrained average-AoI belief MDP in \eqref{ori_prob} as a parameterized unconstrained average cost belief MDP using Lagrangian approach. Given Lagrange multiplier $\lambda$, the instantaneous Lagrangian cost at time slot $t$ is defined by
\begin{align}
	C(\mathbf{s}_{t},u_{t};\lambda)=C_\Delta(\mathbf{s}_{t},u_{t})+\lambda C_E(\mathbf{s}_{t},u_{t}).
	\label{inst_lag_cost}
\end{align}
Then, the average Lagrangian cost under policy $\pi$ is given by
\begin{align}
	&\bar{L}(\pi;\lambda)= \limsup_{T \rightarrow \infty} \frac{1}{T}\mathbb{E_\pi}\big[\sum_{t=1}^T C(\mathbf{s}_{t},u_{t}; \lambda)\big].
\end{align}

Then, we have an unconstrained average cost belief MDP which aims at minimizing the above average Lagrangian cost:

\textit{Problem 2 (Unconstrained average cost belief MDP):}
\begin{align}
	\bar{L}^*(\lambda) \triangleq \min_{\pi} \ \bar{L}(\pi;\lambda)
	\label{unconstrained_avg},
\end{align}
where $\bar{L}^*(\lambda)$ is the optimal average Lagrangian cost with regard to $\lambda$. A policy is said to be \emph{average cost optimal} if it minimizes the average Lagrangian cost.

The relation between the optimal solutions of the problems \eqref{ori_prob} and \eqref{unconstrained_avg} is provided in the following corollary.
\begin{corollary}
	The optimal average AoI of problem \eqref{ori_prob} and the optimal average Lagrangian cost of problem \eqref{unconstrained_avg} satisfy
	\begin{align}
	\vspace{-0.2cm}
		\bar{A}^*=\sup_{\lambda\geq0}\bar{L}^*(\lambda)-\lambda E_{\text{max}}
	\end{align}
\end{corollary}
\begin{proof}
	By Theorem 12.7 in \cite{altman1999constrained}, we only need to check the following condition:
		for all $r \in \mathbb{R}$, the set $G(r)\triangleq\{\mathbf{s}\in \mathcal{S}:\inf_u C_\Delta (\mathbf{s},u)<r\}$ is finite. Given $r$, for any $\mathbf{s}=(\Delta,k,\omega)\in G(r)$, $\Delta=\inf_u C_\Delta (\mathbf{s},u)<r$. With fixed finite $\Delta$, $\Omega_\Delta$ is finite. Thus, $G(r)$ is finite.
\end{proof}

\section{Structure Based Algorithm Design}
\label{structure}  
In this section, we investigate the structure of the optimal policy for the constrained average-AoI belief MDP in \eqref{ori_prob} and propose a structure-aware algorithm. 
\subsection{Structure of Constrained Average-AoI Optimal Policy}
\label{threshold_type_policy}
\subsubsection{Main results}
To explore the structure, we first show that there exists a stationary deterministic \emph{threshold-type} scheduling policy that solves the unconstrained average cost belief MDP in \eqref{unconstrained_avg}.   
\begin{theorem}
\label{thre_avg}
	Given $\lambda$, there exists a stationary deterministic unconstrained average cost optimal policy that is of threshold-type in belief. Specifically, \eqref{unconstrained_avg} can be minimized by a policy of the form $\pi_\lambda^\star=(d_\lambda^\star,d_\lambda^\star,\cdots)$, where 
\begin{align}
\label{avg_opt_policy}
 &d_\lambda^\star(\Delta,k,\omega)=
  \begin{cases}
  0 & \text{if}\, \, 0\leq \omega<\omega^\star(\Delta,k;\lambda), \\
  1 & \text{if} \, \,\omega^\star(\Delta,k;\lambda)\leq \omega,
  \end{cases} 
\end{align}
 where $\omega^\star(\Delta,k;\lambda)$ denotes the threshold given pair of AoI and relative slot index $(\Delta, k)$ and Lagrange multiplier $\lambda$.
\end{theorem}
\begin{proof}
Please see Section \ref{analysis}.
\end{proof}

\noindent Note that the techniques in papers dealing with threshold property in POMDP \cite{lovejoy1987some, albright1979structural,laourine2010betting,liu2010indexability,abad2017channel} cannot be applied to our problem. This is because, given hidden state and action, the one-stage cost in these papers is \emph{constant and bounded}, while the one-stage cost in our paper depends on \emph{varying and unbounded} AoI. Next, we show that the optimal policy for the original problem \eqref{ori_prob} is a mixture of no more than two stationary deterministic threshold-type policies.
\begin{corollary}
\label{opt_policy_sturcture}
	There exists a stationary randomized policy $\pi^\star$ that is the optimal solution to the constrained average-AoI belief MDP in \eqref{ori_prob}, where $\pi^\star$ is a randomized mixture of threshold-type policies as follows:
\begin{align}
	\pi^\star=q \pi_{\lambda_1}^\star+(1-q) \pi_{\lambda_2}^\star,
	\label{randomized_mixture}
\end{align}
where $q \in [0,1]$ is a randomization factor, and $\pi_{\lambda_1}^\star$ and $\pi_{\lambda_2}^\star$ are the optimal threshold-type policies \eqref{avg_opt_policy} for some Lagrange multipliers $\lambda_1$ and $\lambda_2$, respectively.
\end{corollary}
\begin{proof}
	Note that a stationary policy that transmits at the beginning of every $\left \lceil{\frac{1}{KE_{\text{max}}}}\right \rceil$ frames satisfies energy constraint, where $\lceil \cdot\rceil$ is the ceil function. Thus, the problem \eqref{ori_prob} is feasible. Together with our unichain assumption, the result follows from Theorem 4.4 in \cite{altman1999constrained}.
\end{proof}
The method to determine $\lambda_1$, $\lambda_2$ and $q$ will be discussed in Section \ref{lagrange}.

\subsubsection{Proof of Theorem \ref{thre_avg}}
\label{analysis}
We prove Theorem \ref{thre_avg} in two steps: (i) address an unconstrained discounted cost belief MDP; (ii) relate it to the unconstrained average cost belief MDP. In particular, we show that the optimal policy for the unconstrained discounted cost belief MDP is of threshold-type in $\omega$, which implies that the optimal policy for the unconstrained average cost belief MDP is of threshold-type in $\omega$

Given an initial state $\mathbf{s}$, the total expected discounted Lagrangian cost under policy $\pi$ is given by
\begin{align}
\vspace{-0.2cm}
	&L_{\mathbf{s}}^\beta(\pi;\lambda)= \limsup_{T \rightarrow \infty} \mathbb{E_\pi}\big[\sum_{t=1}^T \beta^{t-1}C(\mathbf{s}_{t},u_{t};\lambda)|\mathbf{s}\big],
	\label{lag_disc_cost}
\end{align}
where $\beta \in (0, 1)$ is a discount factor. The optimization problem of minimizing the total expected discounted Lagrangian cost can be cast as

\textit{Problem 3 (Unconstrained discounted cost belief MDP):}
\begin{equation}
	V^\beta(\mathbf{s})\triangleq\min_{\pi} \ L_{\mathbf{s}}^\beta(\pi;\lambda),
	\label{disc_cost_opt}
\end{equation}
where $V^\beta(\mathbf{s})$ denotes the optimal total expected $\beta$-discounted Lagrangian cost (for convenience, we omit $\lambda$ in notation $V^\beta(\mathbf{s})$).

A policy is said to be \emph{$\beta$-discounted cost optimal} if it minimizes the total expected $\beta$-discounted Lagrangian cost. 
In Proposition \ref{existence_discount}, we introduce the optimality equation of $V^\beta(\mathbf{s})$.
\begin{proposition}
\label{existence_discount}
	\noindent	(a) The optimal total expected $\beta$-discounted Lagrangian cost $V^\beta(\Delta,k,\omega)$ satisfies the optimality equation as follows:	
	\begin{align}
		V^\beta\left(\Delta,k,\omega\right)=\min_{u\in \{0,1\}}  Q^\beta\left(\Delta,k,\omega;u\right), \label{opt_equ_disc}
		\vspace{-0.3cm}
	\end{align}
	\vspace{-0.2cm}	
	where
	\begin{align}
		Q^\beta\left(\Delta,k,\omega;0\right)=&\Delta+\beta V^\beta\left(\Delta+1,\left(k\right)_{+},\mathcal{T}\left(\omega\right)\right);\label{q1_}\\
		Q^\beta\left(\Delta,k,\omega;1\right)=&\Delta+\lambda
		+\beta \Big(\omega V^\beta\left(k, \left(k\right)_{+},p_{11}\right)\notag\\+&(1-\omega)V^\beta\left(\Delta+1,\left(k\right)_{+},p_{01}\right)\Big).\label{q2_}
		\vspace{-0.2cm}
	\end{align}
\noindent(b) A stationary deterministic policy determined by the right-hand-side of \eqref{opt_equ_disc} is $\beta$-discounted cost optimal.
	
\noindent (c) Let $V_n^\beta(\mathbf{s})$ be the cost-to-go function such that $V_0^\beta(\mathbf{s})\!=\!0$, for all $\mathbf{s}\in\mathcal{S}$ and for $n\geq 0$,
		\begin{align}
		V_{n+1}^\beta(\Delta,k,\omega)=\min_{u\in \{0,1\}} Q_{n+1}^\beta(\Delta,k,\omega;u),
		\label{iteration}
	    \end{align}
	    where
	    \vspace{-0.2cm}
	    \begin{align}
		Q_{n+1}^\beta\left(\Delta,k,\omega;0\right)=&\Delta+\beta V_{n}^\beta\left(\Delta+1,\left(k\right)_{+},\mathcal{T}\left(\omega\right)\right);\label{q1}\\
		Q_{n+1}^\beta\left(\Delta,k,\omega;1\right)=&\Delta+\lambda+\beta \Big(\omega V_{n}^\beta\left(k, \left(k\right)_{+},p_{11}\right)\notag\\+&(1-\omega)V_{n}^\beta
		\left(\Delta+1,\left(k\right)_{+},p_{01}\right)\Big).\label{q2}
		\vspace{-0.2cm}
	    \end{align}
	    Then, we have $V_n^\beta(\mathbf{s}) \rightarrow V^\beta(\mathbf{s})$ as $n\rightarrow \infty$, for every $\mathbf{s}$, $\beta$.
\end{proposition}
\begin{proof}
According to \cite{sennott1989average}, it suffices to show that there exists a stationary deterministic policy $f$ such that for all $\beta, \mathbf{s}$, we have $L_{\mathbf{s}}^\beta(f;\lambda)\!<\!\!\infty$. Let $f$ be a policy that chooses $u=0$ for every time slot. For any initial state $\mathbf{s}_1=(\Delta,t,\omega)$ under this policy, we have
\begin{align}
	L_{\mathbf{s}_1}^\beta(f;\lambda)&= \limsup_{T \rightarrow \infty} \mathbb{E}_f\big[\sum_{t=1}^T \beta^{t-1}C(\mathbf{s}_{t},0;\lambda)|\mathbf{s}_{1}\big]\notag\\
	&=\sum_{n=0}^{\infty}\beta^n(\Delta+n)\notag\\
	&=\frac{\Delta}{1-\beta}+\frac{\beta}{(1-\beta)^2}<\infty.\notag
\end{align}
\end{proof}

Using (c) in Proposition \ref{existence_discount}, we show properties of $V^\beta$ in Lemma \ref{property}.
\begin{lemma}
\label{property}
   If $p_{11}\geq p_{01}$, then the value function $V^\beta$ has the following properties:
   
\noindent(a) $V^\beta(\Delta, k, \omega)$ is non-decreasing with regard to age $\Delta$.

\noindent(b) $V^\beta(\Delta, k, \omega)$ is non-increasing with regard to belief $\omega$.

\noindent(c) For beliefs $x, y, z,\omega$ that satisfy $z=\omega x+(1-\omega)y$ and $x\geq y$, we have 
\begin{equation}
\!(\!1\!-\!\omega)\lambda\!+\omega V^\beta(\Delta, k, x)\!+\!(1\!-\!\omega)V^\beta(\Delta, k, y)\!\geq \!\!V^\beta(\Delta, k, z)\!.\label{prop 2}
\end{equation}

\noindent(d) The optimal policy corresponding to $V^\beta$ is of a threshold-type in $\omega$, i.e. given $\Delta$, $k$, there exists a threshold $\omega_{\beta}^*(\Delta, k;\lambda)$ such that it is optimal to transmit only when $\omega\geq\omega_{\beta}^*(\Delta, k;\lambda)$.
\end{lemma}
\begin{proof}
Please see Appendix \ref{P2}.
\end{proof}

By (d) in Lemma \ref{property}, the $\beta$-discounted cost optimal policies are of threshold-type in belief. 
By \cite{sennott1989average}, under certain conditions (A proof of these conditions verification is provided in Appendix \ref{exist_1}), average cost optimal policy can be viewed as a limit of a sequence of $\beta$-discounted cost optimal policies as $\beta\rightarrow 1$. Thus, the average cost optimal policies are of threshold-type in belief.	

\subsection{Structure-Aware Algorithm Design}
\label{policy}
We exploit Corollary \ref{opt_policy_sturcture} to design a structure-aware algorithm for \eqref{ori_prob} in two steps: We first design a structure-aware algorithm for \eqref{unconstrained_avg}, and then construct a way to determine parameters $\lambda_1$, $\lambda_2$ and $q$. 

\subsubsection{Structure-Aware Algorithm for the approximate unconstrained average cost belief MDP}
\label{approximate_algo}
In practice, classic value iteration cannot work if state space is infinite. To deal with this, we first propose a finite-state approximation for infinite-state belief MDP in \eqref{unconstrained_avg} and show the convergence of our approximate belief MDPs to the original one.

Let $N$ be an upper bound for the AoI and the number of Markov transitions from $p_{01}$ or $p_{11}$. Since $\mathcal{T}^i(p_{01})\leq\mathcal{T}^{i+1}(p_{01})$ and $\mathcal{T}^i(p_{11})\geq\mathcal{T}^{i+1}(p_{11})$ for $i\in \mathbb{N}$, we have that with bound $N$, the state space of the approximate belief MDP is given by $\mathcal{S}^N\!\triangleq\!\{(\Delta,k,\omega)\!\in \!\mathcal{S}\!:\! \Delta\leq N, p_{01}\leq \omega\!\leq\!\mathcal{T}^N(p_{01}) \ \text{or}\  \mathcal{T}^N(p_{11})\leq \omega \leq p_{11} \}$. Without loss of generality, we assume $N>K$.

Given the state $(\Delta_t, k_t, \omega_t)\in \mathcal{S}^N$, the state $\mathbf{s}_{t+1}=(\Delta_{t+1}, k_{t+1}, \omega_{t+1})\in \mathcal{S}^N$ is updated as follows:
\begin{align}
 &\mathbf{s}_{t+1} 
  \!\!=\!\!\begin{cases}
  \left(k_t,\left(k_t\right)_+\!,p_{11}\right) \! &\text{if} \, u_{t}\!=\!1, \theta_{t}\!=\!1, \\  
  \left(\phi(\Delta_{t}\!+\!1),\left(k_t\right)_+\!,p_{01}\right)\!  &\text{if} \, u_{t}\!=\!1, \theta_{t}\!=\!0, \\
  \left(\phi(\Delta_{t}\!+\!1),\left(k_t\right)_+\!,\varphi(\mathcal{T}\left(\omega_{t}\right))\right)\!\!\! & \text{if} \, u_{t}\!=\!0,
  \end{cases} 
  \label{omega2}
\end{align}
where $\phi(x)=\min\{x,N\}$, and $\varphi(y)$ is given by\footnotemark[1]
\begin{align}
 &\varphi(y)=
  \begin{cases}
  \mathcal{T}^N(p_{11}) & \text{if}\, \, \mathcal{T}^N(p_{01}) <y <\mathcal{T}^N(p_{11}), \\
  y & \text{otherwise}.
  \end{cases} 
\end{align}
\footnotetext[1]{{We upper bound the belief state by $\mathcal{T}^N(p_{11})$. This ensures that the optimal policy for the approximate unconstrained belief MDP is of threshold-type.}}

Given action $u$, the transition probability from $\mathbf{s}$ to $\mathbf{s}'$ on state space $\mathcal{S}^N$, denoted by $P^{N}_{\mathbf{s}\mathbf{s}'}(u)$, is expressed as 
\begin{align}
		P^{N}_{\mathbf{s}\mathbf{s}'}(u)=P_{\mathbf{s}\mathbf{s}'}(u)+\sum_{\mathbf{r}\in \mathcal{S}-\mathcal{S}^{N}}P_{\mathbf{s}\mathbf{r}}(u)\mathbbm{1}_{\{\nu(\mathbf{r})=\mathbf{s}'\}},
\end{align}
where $P_{\mathbf{s}\mathbf{s}'}(u)$ and $P_{\mathbf{s}\mathbf{r}}(u)$ are the transition probabilities on $\mathcal{S}$ defined in \eqref{transition_prob}, $\mathbbm{1}_{\{\cdot\}}$ is the indicator function, and approximation operation to state is 
\begin{equation}
	\nu\left(\left(z1,z2,z3\right)\right)\triangleq(\phi(z1),z2,\varphi(z3)).\label{approximation_opr}
\end{equation}

In general, a sequence of approximate MDPs may not converge to the original MDP \cite{sennott2009stochastic}. In Theorem \ref{approx}, we show the convergence of our approximate MDPs to the original MDP. 
\begin{theorem}
\label{approx}
	Let $\bar{L}^{N*}(\lambda)$ be the minimum average Lagrangian cost for the approximate MDP with regard to bound $N$ and Lagrange multiplier $\lambda$. Then, $\bar{L}^{N*}(\lambda) \rightarrow \bar{L}^*(\lambda)$ as $N\rightarrow \infty$.
\end{theorem}
\begin{proof}
	Please see Appendix \ref{convg}.
\end{proof}

The Relative Value Iteration (RVI) algorithm can be utilized to obtain an optimal stationary deterministic policy for the approximate MDP. In particular, RVI starts with $V_0^N(\mathbf{s})=0$, $\forall \mathbf{s}\in \mathcal{S}^N$ and updates $V_{n+1}^N(\mathbf{s})$ by minimizing the RHS of equation \eqref{iter_avg} in the $(n+1)$-th iteration, $n\in \{0,1,2,\cdots\}$.
\vspace{-0.1cm}
\begin{align}
	V_{n+1}^{N}(\mathbf{s})=\min_{u}\Big\{&C(\mathbf{s},u;\lambda)\notag\\
	&+\sum_{\mathbf{s}'\in \mathcal{S}^{N}}P^{N}_{\mathbf{s}\mathbf{s}'}(u)h_n^{N}(\mathbf{s}')-h_n^{N}(\mathbf{0})\Big \},
	\label{iter_avg}
\end{align}
where $\mathbf{0}$ is the reference state and $h^{N}_n(\mathbf{s})=V^{N}_n(\mathbf{s})-V^{N}_n(\mathbf{0})$. 
Note that similar to the proof in Section \ref{threshold_type_policy}, it can be shown that the optimal policy for the approximate MDP is still of threshold-type. Thus, we utilize the threshold property in RVI algorithm and propose a threshold-type RVI to reduce the complexity in Algorithm \ref{alg1} (Line 4-24). For each iteration, we update the threshold $\omega^\star(\Delta,k;\lambda)$ (Line 16) in addition to $V^N(\mathbf{s})$. If certain state satisfies the threshold condition (Line 11), then the optimal action for the state in this iteration is determined immediately without doing the optimization operation (Line 12), which reduces the algorithm complexity.

\begin{algorithm}
\caption{Structure-Aware Scheduling without channel sensing}
\label{alg1}
\LinesNumbered
\footnotesize given tolerance $\epsilon>0, \epsilon_\lambda>0$, $\lambda^{*-}$, $\lambda^{*+}$, $N$ \;
\While{$|\lambda^{*+}-\lambda^{*-}|>\epsilon_\lambda$}{
$\lambda=(\lambda^{*+}+\lambda^{*-})/2$\;
$V^{N}(\mathbf{s})=0,h^{N}(\mathbf{s})=0,h^{N}_\textit{prev}(\mathbf{s})=\infty,$ for all $\mathbf{s}\in \mathcal{S}^{N}$\;
 \While{$\max_{\mathbf{s}\in \mathcal{S}^{N}}|h^{N}(\mathbf{s})-h^{N}_{\textit{prev}}(\mathbf{s})|>\epsilon$}{
$ \omega^*(\Delta,k;\lambda)=\infty$ for all $\mathbf{s}=(\Delta,k,\omega)\in \mathcal{S}^{N}$\;
  \ForEach{$\mathbf{s}=(\Delta,k,\omega)\in \mathcal{S}^{N}$}{
  \eIf{$\Delta<K$}{$u^*=0$\;
       }{
       \eIf{$\omega\geq \omega^*(\Delta,k;\lambda)$}{
        $u^*=1$\;
       }{
       $u^*=\argmin_{u\in \{0,1\}} \{C(\mathbf{s},u;\lambda)+\sum_{\mathbf{s}'\in \mathcal{S}^{N}}P^{N}_{\mathbf{s}\mathbf{s}'}(u)h^{N}(\mathbf{s}')\}$\;
         \If{$u^*=1$}{
   $\omega^*(\Delta,k;\lambda)=\omega$\;
           }
       }
       $V^{N}(\mathbf{s})=C(\mathbf{s},u^*;\lambda)+\sum_{\mathbf{s}'\in \mathcal{S}^{N}}P^{N}_{\mathbf{s}\mathbf{s}'}(u^*)h^{N}(\mathbf{s}')-h^{N}(\mathbf{0})$\;
   }
  $h^{N}_\textit{prev}(\mathbf{s})=h^{N}(\mathbf{s})$\;
  $h^{N}(\mathbf{s})=V^{N}(\mathbf{s})-V^{N}(\mathbf{0})$\;
 }
 }
 Compute the average energy cost $\bar{E}(\lambda)$\;
 \eIf{$\bar{E}(\lambda)>E_{\text{max}}$}{
 $\lambda^{*-}=\lambda$\;}{
 $\lambda^{*+}=\lambda$\;
 }
 }
\end{algorithm}

\subsubsection{Lagrange Multiplier Estimation}
\label{lagrange}
By Lemma 3.4 of \cite{sennott1993constrained}, for $\lambda_1\!\!<\!\!\lambda_2$, we have $\bar{A}(\pi^\star_{\lambda_1})\!\!\leq\!\! \bar{A}(\pi^\star_{\lambda_2})$ and $\bar{E}(\pi^\star_{\lambda_1})\!\!\geq \!\!\bar{E}(\pi^\star_{\lambda_2})$. Thus, the optimal Lagrangian multiplier $\lambda^\star$ is defined as $\lambda^\star\!\triangleq\!\! \inf \{\lambda\!\!>\!\!0\!:\!\bar{E}(\pi_{\lambda}^\star)\leq E_{\text{max}}\}$. If there exists $\lambda^\star$ such that $\bar{E}(\pi^\star_{\lambda^\star})= E_{\text{max}}$, then the constrained average-AoI optimal policy is a stationary deterministic policy where $q$ in Corollary \ref{opt_policy_sturcture} is either 0 or 1. Otherwise,
the optimal policy $\pi^\star$ chooses policy $\pi^\star_{\lambda^{\star-}}$ with probability $q$ and policy $\pi^\star_{\lambda^{\star+}}$ with probability $1-q$. The randomization factor $q$ can be computed by
\begin{equation}
\label{q}
	q=\frac{E_{\text{max}}-\bar{E}(\pi^\star_{\lambda^{\star+}})}{\bar{E}(\pi^\star_{\lambda^{\star-}})-\bar{E}(\pi^\star_{\lambda^{\star+}})}.
\end{equation}
The bisection method is used to compute $\lambda^{\star-}$, $\lambda^{\star+}$ and thus $q$ (Line 2-3 and Line 26-30 in Algorithm \ref{alg1}). The algorithm starts with $\lambda^{\star-}=0$ and sufficiently large $\lambda^{\star+}$. 

\section{Scheduling with Delayed Channel Sensing}
\label{case22}
With delayed channel sensing, the CSI of the last time slot is always available at the beginning of each slot. Thus, the problem in this case can be formulated as a constrained MDP. The state space reduces to $\mathcal{S}\!\triangleq \!\{(\Delta,k,g)\!:\!k\in \mathcal{K}, \Delta\in \mathcal{A}_k,  g\in \{0,1\}\}$, where $g$ denotes the CSI of the last time slot. Given $\mathbf{s}_t=(\Delta_t,k_t,g_t)$ and $u_t$ at time slot $t$, the transition probability to $\mathbf{s}_{t+1}=(\Delta_{t+1},k_{t+1},g_{t+1})$ is written as follows:
  \vspace{-0.15cm}	
\begin{align}
 &P_{\mathbf{s}_{t}\mathbf{s}_{t+1}}(u_{t})\notag\\=
  &\begin{cases}
  p_{g_t1} & \text{if} \, u_{t}=1,\mathbf{s}_{t+1}=(k_t, (k_t)_+,1), \\
  1-p_{g_t1} & \text{if}  \, u_{t}=1,\mathbf{s}_{t+1}=(\Delta_t+1, (k_t)_+,0),\\  
  1 &\text{if}  \, u_{t}=0,\mathbf{s}_{t+1}=(\Delta_t+1, (k_t)_+,g_{t+1}).
  \end{cases} 
\end{align}

Following Section \ref{Lag_transform} and Section \ref{structure}, the optimal transmission scheduling policy in this case is also a randomized mixture of no more than two deterministic policies, each of which is optimal for an unconstrained average cost MDP. But thanks to the simplification in state, we can show that the optimal policy for the unconstrained average cost MDP in this case is of threshold-type in AoI in Theorem \ref{thre_avg2}.
\begin{theorem}
\label{thre_avg2}
	Given Lagrange multiplier $\lambda$, there exists a stationary unconstrained average cost optimal policy that is deterministic and of threshold-type in AoI. Specifically, the policy is in the form $\pi_\lambda^*=(d_\lambda^*,d_\lambda^*,\cdots)$, where
\begin{align}
\label{avg_opt_policy2}
 &d_{\lambda}^*(\Delta,k,g)=
  \begin{cases}
  0 & \text{if}\, \, 0\leq \Delta<\Delta^*(k,g;\lambda), \\
  1 & \text{if} \, \,\Delta^*(k,g;\lambda)\leq \Delta,
  \end{cases} 
\end{align}
and
\vspace{-0.2cm}
\begin{equation}
	\Delta^*(k,1;\lambda)\leq\Delta^*(k,0;\lambda),
	\label{thre_relation}
	\vspace{-0.2cm}
\end{equation}
 where $\Delta^*(k,g;\lambda)$ denotes the threshold given pair of relative slot index and delayed CSI $(k,g)$ and Lagrange multiplier $\lambda$.
\end{theorem}
Different from Theorem \ref{thre_avg} which provides threshold structure in belief $\omega$, Theorem \ref{threshold_prop2} obtains that (i) the average cost optimal policy is of threshold-type in AoI, and (ii) threshold when $g=1$ is no larger than the threshold when $g=0$. Indeed, (ii) is used in algorithm to further reduce algorithm complexity. In particular, similar to Section \ref{approximate_algo}, we bound AoI with $N$ and propose a threshold-type algorithm in Algorithm \ref{alg2} to minimize unconstrained average cost. Different from corresponding part in Algorithm \ref{alg1}, $\Delta^*(k,1;\lambda)$ is updated along with each threshold updating (Line 15) to keep the threshold relation in \eqref{thre_relation}. This further reduces algorithm complexity.

\begin{algorithm}[]
\caption{Threshold-type scheduling for unconstrained average cost MDP with delayed channel sensing}
\label{alg2}
\LinesNumbered
\footnotesize given tolerance $\epsilon>0$, Lagrange multiplier $\lambda$ and bound $N$ \;
$V^{N}(\mathbf{s})=0,h^{N}(\mathbf{s})=0,h^{N}_\textit{prev}(\mathbf{s})=\infty,$ for all $\mathbf{s}\in \mathcal{S}^{N}$\;
 \While{$\max_{\mathbf{s}\in \mathcal{S}^{N}}|h^{N}(\mathbf{s})-h^{N}_{\textit{prev}}(\mathbf{s})|>\epsilon$}{
$ \Delta^*(k,g;\lambda)=\infty$ for all $\mathbf{s}=(\Delta,k,g)\in \mathcal{S}^{N}$\;
  \ForEach{$\mathbf{s}=(\Delta,k,g)\in \mathcal{S}^{N}$}{
  \eIf{$\Delta<K$}{$u^*=0$\;
       }{
       \eIf{$\Delta\geq \Delta^*(k,g;\lambda)$}{
        $u^*=1$\;
       }{
       $u^*=\argmin_{u\in \{0,1\}} \{C(\mathbf{s},u;\lambda)+\sum_{\mathbf{s}'\in \mathcal{S}^{N}}P^{N}_{\mathbf{s}\mathbf{s}'}(u)h^{N}(\mathbf{s}')\}$\;
         \If{$u^*=1$}{
   $\Delta^*(k,g;\lambda)=\Delta$\;
   $\Delta^*(k,1;\lambda)=\min\{\Delta,\Delta^*(k,1;\lambda)\}$\;
           }
       }
       $V^{N}(\mathbf{s})=C(\mathbf{s},u^*;\lambda)+\sum_{\mathbf{s}'\in \mathcal{S}^{N}}P^{N}_{\mathbf{s}\mathbf{s}'}(u^*)h^{N}(\mathbf{s}')-h^{N}(\mathbf{0})$\;
   }
  $h^{N}_\textit{prev}(\mathbf{s})=h^{N}(\mathbf{s})$\;
  $h^{N}(\mathbf{s})=V^{N}(\mathbf{s})-V^{N}(\mathbf{0})$\;
 }
 }
 \end{algorithm}

The proof idea of Theorem \ref{threshold_prop2} is similar to Theorem \ref{thre_avg}. We relate average cost MDPs to discounted cost MDPs.  Next, we explore the structure of discounted cost optimal policies.

The optimality equation in \eqref{opt_equ_disc} is modified as follows:
 \begin{align}
 	&V^\beta \left(\Delta,k,g\right)\!=\!\Delta+\beta\min\Big\{\!\sum_{g'\in\{0,1\}}\!\!p_{gg'}V^\beta\!\left(\!\Delta+1,\left(k\right)_{+}\!,g'\right),\notag\\
 	&\   \ \lambda+ p_{g1} V^\beta\left(k, \left(k\right)_{+},1\right)\!+\!p_{g0}V^\beta\left(\Delta+1,\left(k\right)_{+},0\right)\Big\}.
 \end{align}
\vspace{-0.3cm}
 
 First, we prove the monotonicity of value function $V^\beta$ in AoI in Lemma \ref{property2}.
\begin{lemma}
\label{property2}
   The function $V^\beta(\Delta, k, g)$ is non-decreasing with regard to AoI $\Delta$.
\end{lemma}
\begin{proof} Please see Appendix \ref{mono_property_case}.
\end{proof}

With this, we characterize the structure of optimal policy for the unconstrained discounted cost MDP in Lemma \ref{threshold_prop2}. 
\begin{lemma}
\label{threshold_prop2}
	Given $\lambda$ and $\beta$, the optimal policy that minimizes the $\beta$-discounted Lagrangian cost is of threshold-type in AoI $\Delta$, i.e. given $k,g$, there exists a threshold $\Delta_{\beta}^*(k,g;\lambda)$ such that it is optimal to transmit only when $\Delta \geq \Delta_{\beta}^*(k,g;\lambda)$. In addition, $\Delta_{\beta}^*(k,1;\lambda)\leq\Delta_{\beta}^*(k,0;\lambda)$.
\end{lemma}
\begin{proof}
Please see Appendix \ref{threshold_discount2}.
\end{proof}

Similar to the proof of Theorem \ref{thre_avg}, we can extend the result to the unconstrained average cost MDP as in Theorem \ref{thre_avg2}.

\section{Numerical Results}
\label{numerical_results}
In this section, we numerically evaluate the performance of the proposed algorithms. We assume $N=1000$ and obtain all simulation results over $10^5$ time slots.
\subsection{Average AoI Performance}
Fig. \ref{Age_energy} plots the AoI-energy tradeoff with different fading characteristics (different $p_{11}$ and $p_{01}$) for the two cases that we consider in this paper. In this simulation, we set $K=3$. The optimal average AoI with no energy constraint is plotted as a gray dashed line accordingly. When comparing Fig. \ref{age_e_case1} with Fig. \ref{age_e_case2}, it is easy to observe that for fixed energy constraint and pair of $p_{11}$ and $p_{01}$, the average AoI with delayed channel sensing is no larger than that without channel sensing. 

Moreover, the curves in Fig. \ref{age_e_case1} and Fig. \ref{age_e_case2} exhibit the same trend as follows. For each pair of $p_{11}$ and $p_{01}$, average AoI decreases with energy constraint. Note that it is prohibited to transmit delivered status update. Thus, even if there is no energy constraint, obtaining the optimal average AoI does not necessarily imply transmitting at every time slot. This explains why the average AoI achieved by our proposed policies approaches the gray line even when $E_\text{max}\neq 1$. In addition, we can observe that for certain energy constraint, the average AoI decreases with either $p_{11}$ or $p_{01}$. This is due to the fact that increase in either $p_{11}$ or $p_{01}$ results in the increase of steady state probability that channel is in good state.
\begin{figure}[t]
 \subfloat[Without channel sensing]{
 \includegraphics[scale=.2]{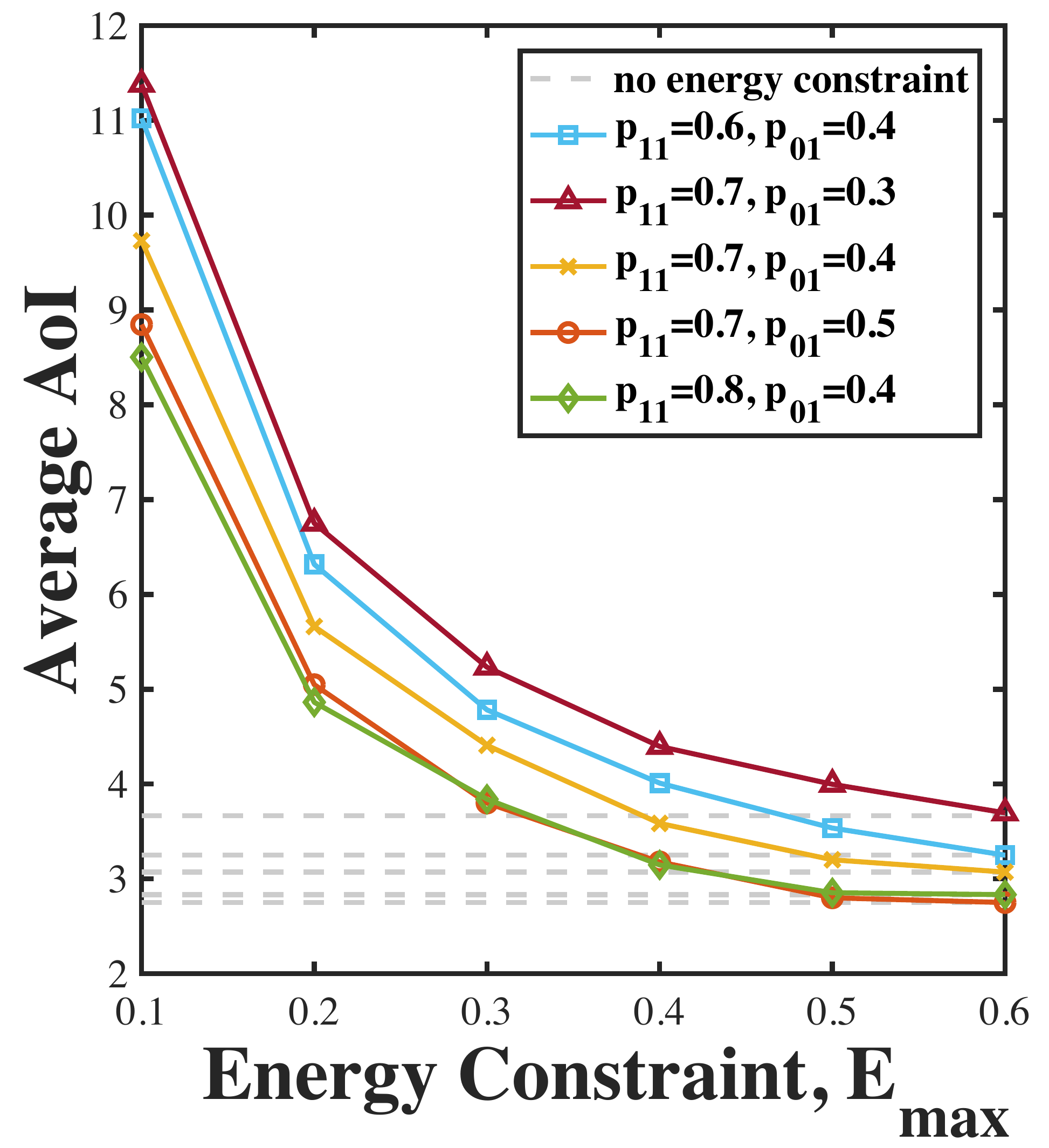}
 \label{age_e_case1}
 }\ 
 \subfloat[With delayed channel sensing]{
 \includegraphics[scale=.2]{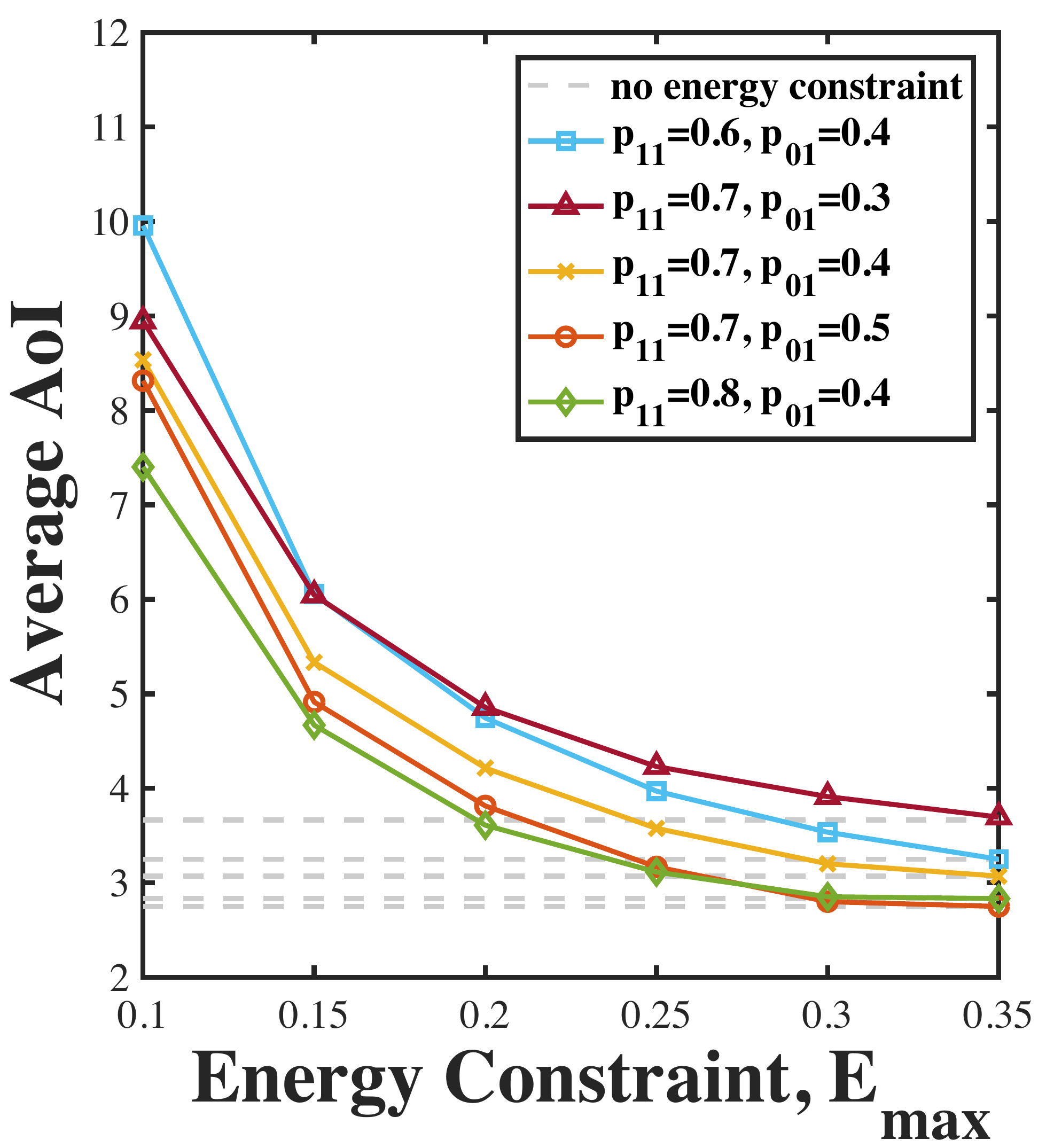}
 \label{age_e_case2}
 }
 \caption{AoI-energy tradeoff with different transition probabilities}
 \label{Age_energy}
 \vspace{-0.5cm}
\end{figure}

Fig. \ref{Age_framelength} studies the average AoI performance vs frame length with different fading characteristics in the two cases. We set the energy constraint $E_{\text{max}}=0.3$. 
\begin{figure}[h]
 \subfloat[Without channel sensing]{
 \includegraphics[scale=.2]{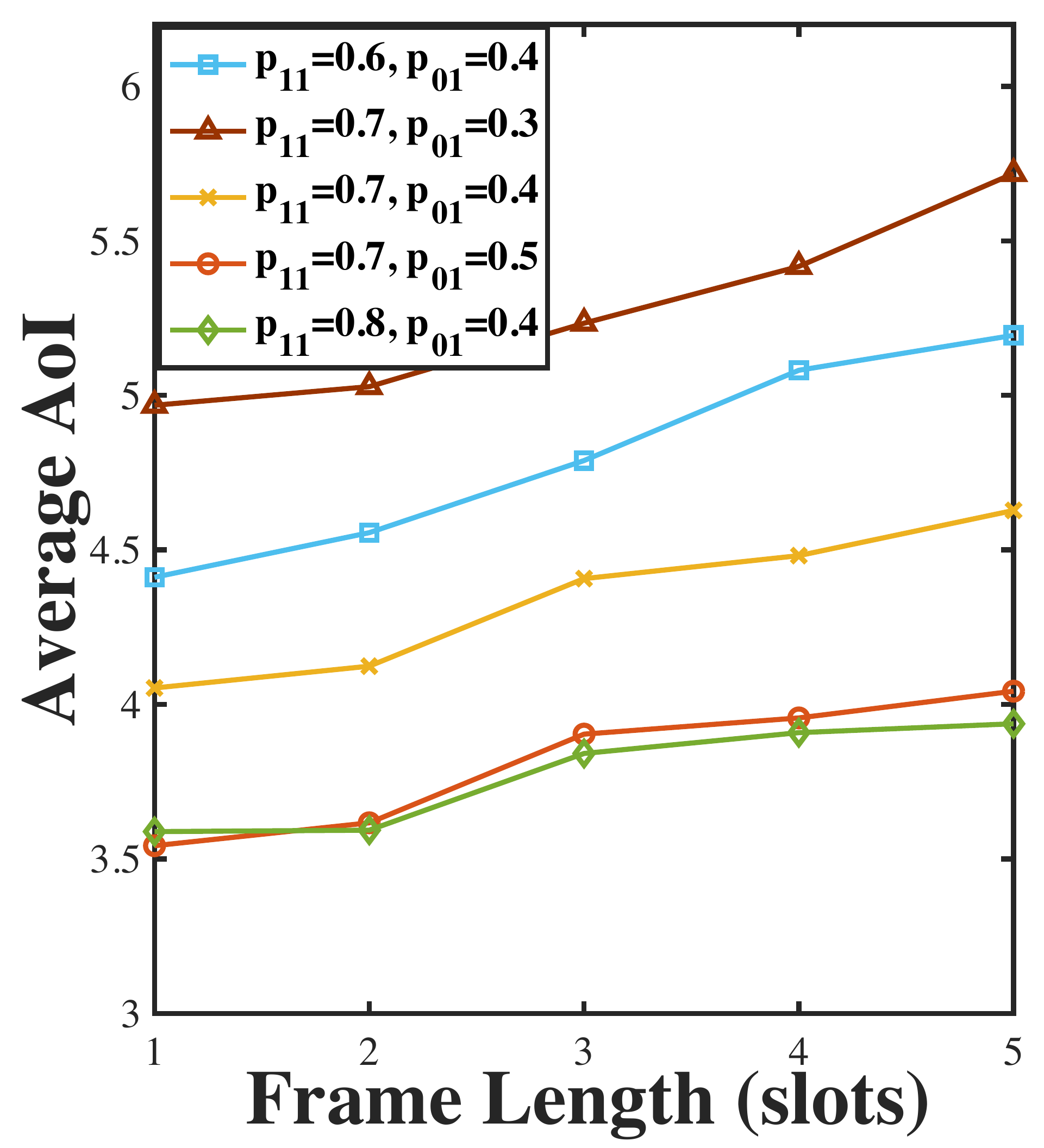}
 \label{age_f_case1}
 }\ \ \
 \subfloat[With delayed channel sensing]{
 \includegraphics[scale=.2]{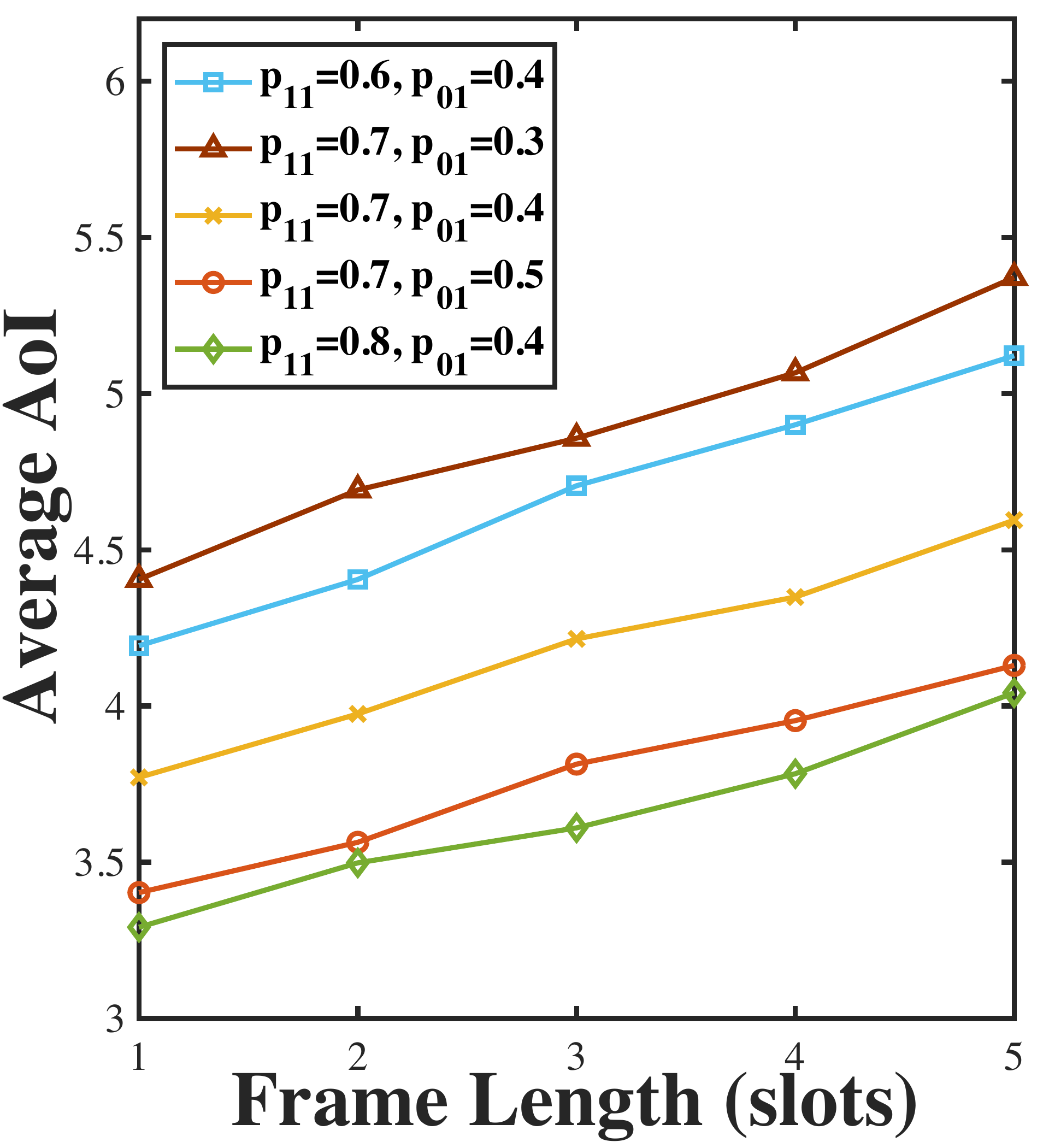}
 \label{age_f_case2}
 }
 \caption{Average AoI vs frame length with different transition probabilities}
	\label{Age_framelength}
\end{figure}

\begin{figure}[h]
	\centering
	\includegraphics[width=0.41\textwidth]{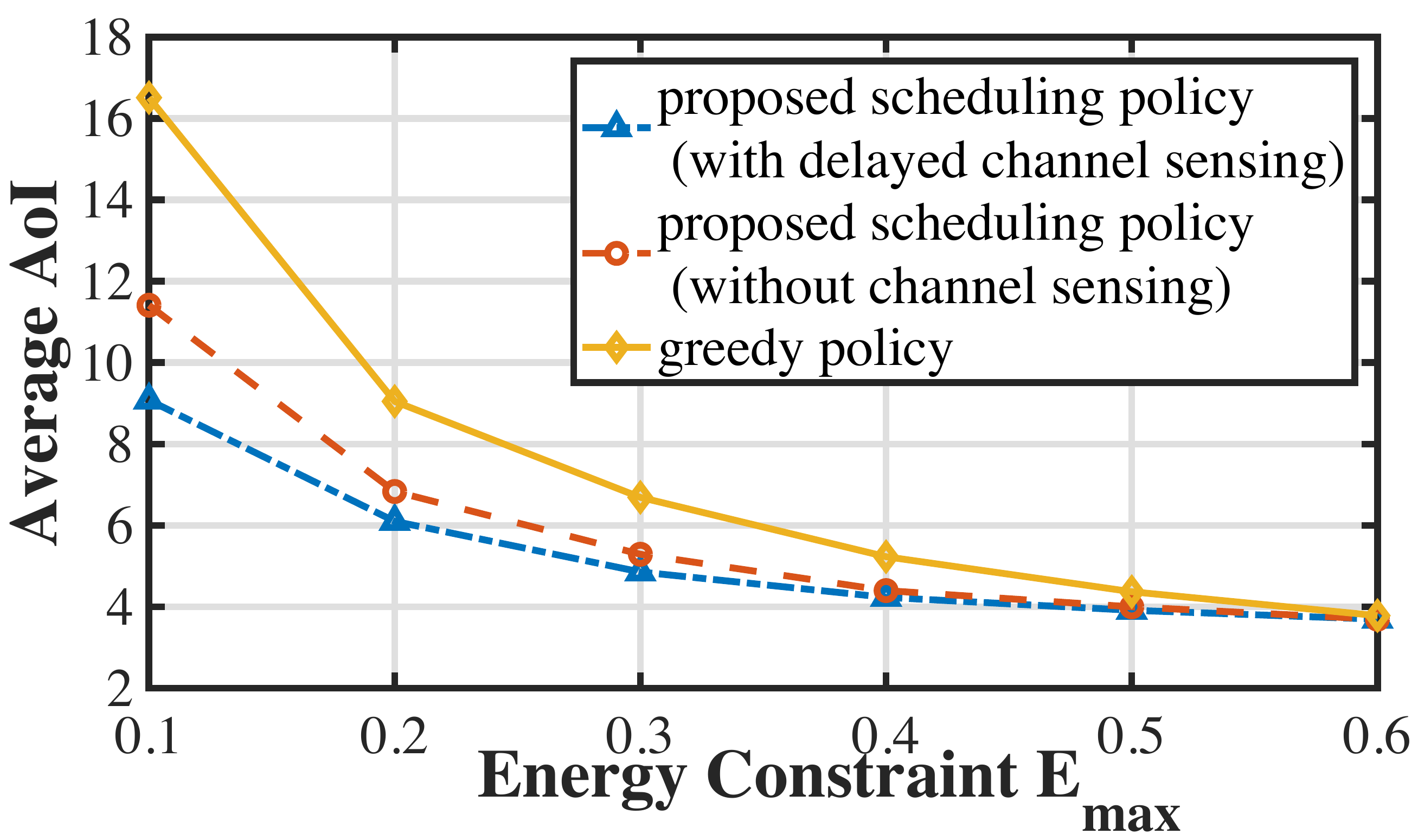}
	\caption{Comparison with greedy policy}
	\label{comp}
\end{figure}

\subsection{Comparison with greedy policy}
Let $e_t$ denote total energy consumption before slot $t$. Then, $\bar{e}_t\triangleq e_t/(t-1)$ denotes the average energy consumed before slot $t$. We compare the proposed transmission scheduling policies with a greedy policy that transmits when $\bar{e}_t<E_{\text{max}}$ and $\Delta_t\geq K$. We set $K=3$, $p_{11}=0.7$, $p_{01}=0.3$, in which case the optimal AoI with no energy constraint is achieved with 0.6167 units energy on average. Thus, the comparison is conducted with energy constraint ranging from 0.1 to 0.6. In Fig. \ref{comp}, it is easy to observe that the proposed transmission scheduling policy outperforms the greedy policy in both cases. The gap between the greedy policy and scheduling policy in either case narrows as the energy constraint is loosened.  

\section{Conclusion}
We studied scheduling transmission of periodically generated updates over a Gilbert-Elliott fading channel in two cases. For the case without channel sensing, the problem is a constrained POMDP and is rewritten as a constrained belief MDP by introducing belief state. We show that the optimal policy for the constrained belief MDP is a randomization of no more than two stationary deterministic policies, each of which is of a threshold-type in the belief on the channel. For the case with delayed channel sensing, we show that the optimal policy has a similar structure as the one in the former case but with AoI associated threshold. In addition, we show that the AoI threshold has monotonic behavior in the delayed channel state in this case. The structure is utilized in either case to reduce algorithm complexity.
\appendices
\section{Proof of Lemma \ref{property}}
\label{P2}

Without loss of generality, we extend space of belief state to $[0, 1]$ and show that (a)-(d) hold for $\omega\in [0,1]$. By Proposition \ref{existence_discount}, $V_n^\beta(\mathbf{s})\! \rightarrow \!V^\beta(\mathbf{s})$ as $n\rightarrow \!\!\infty$. Thus, we show that $V_n^\beta(\mathbf{s})$ satisfies (a)-(d) for $n\geq 0$ via induction. Note that $V_0^\beta(\mathbf{s})=0$ satisfies (a)-(d).  

Suppose that (a)-(d) hold for $n$. We (1) show that (d) holds for $n+1$ based on the assumption that (a)-(c) hold for $n$, and (2) show that (c) holds for $n+1$ based on the result that (d) hold for $n+1$ shown in step (1) and the assumption that (a)-(c) hold for $n$.

\emph{Step (1):} We show that (d) holds for $n\!+\!1$.
Recall that $V_{n+1}^\beta(\mathbf{s})=\min \{Q_{n+1}^\beta(\mathbf{s};1),Q_{n+1}^\beta(\mathbf{s};0)\}$. Thus, we can obtain the threshold property by examining the Q functions $Q_{n+1}^\beta(\mathbf{s};0)$ and $Q_{n+1}^\beta(\mathbf{s};1)$ given in \eqref{q1} and \eqref{q2}. By the expression in \eqref{q2}, $Q_{n+1}^\beta(\Delta,k,\omega;1)$ is linear in $\omega$. Besides, the value function $V_{n}^\beta(\Delta,k,\omega;1)$ in our case is a piecewise linear and concave function with respect to the belief state for all $n$, which can be shown via induction similar to \cite{smallwood1973optimal}. Thus, $Q_{n+1}^\beta(\mathbf{s};0)$ is concave by \eqref{q2}. Moreover, by definition, we have $Q_{n+1}^\beta(\Delta,k,0;1)\!\!\geq\!\! Q_{n+1}^\beta(\Delta,k,0;0)$. 
%
Based on the relation between values of $Q_{n+1}^\beta(\Delta,k,1;1)$ and $Q_{n+1}^\beta(\Delta,k,1;0)$, there are two possible cases for curves of $Q_{n+1}^\beta(\Delta,k,\omega;1)$ and $Q_{n+1}^\beta(\Delta,k,\omega;0)$ as shown in Fig. \ref{CompQ1Q2}.

\begin{figure}[h]
\centering
\subfloat[]{
 \includegraphics[scale=.38]{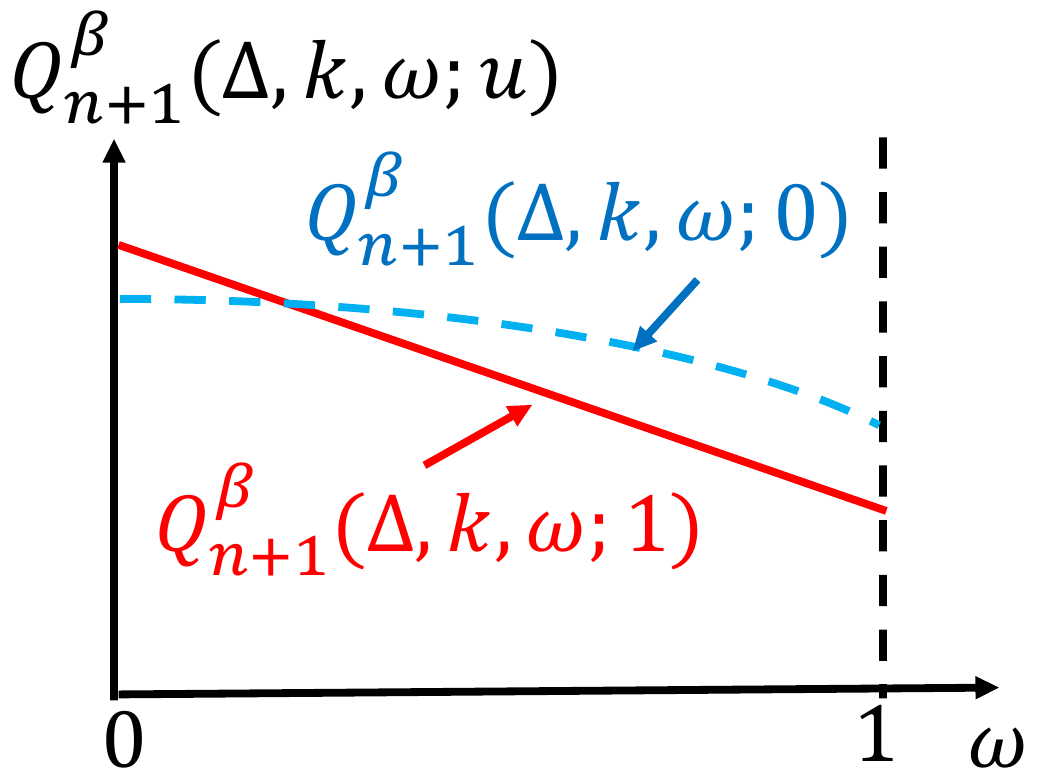}
 \label{case1}
 }
\subfloat[]{
 \includegraphics[scale=.38]{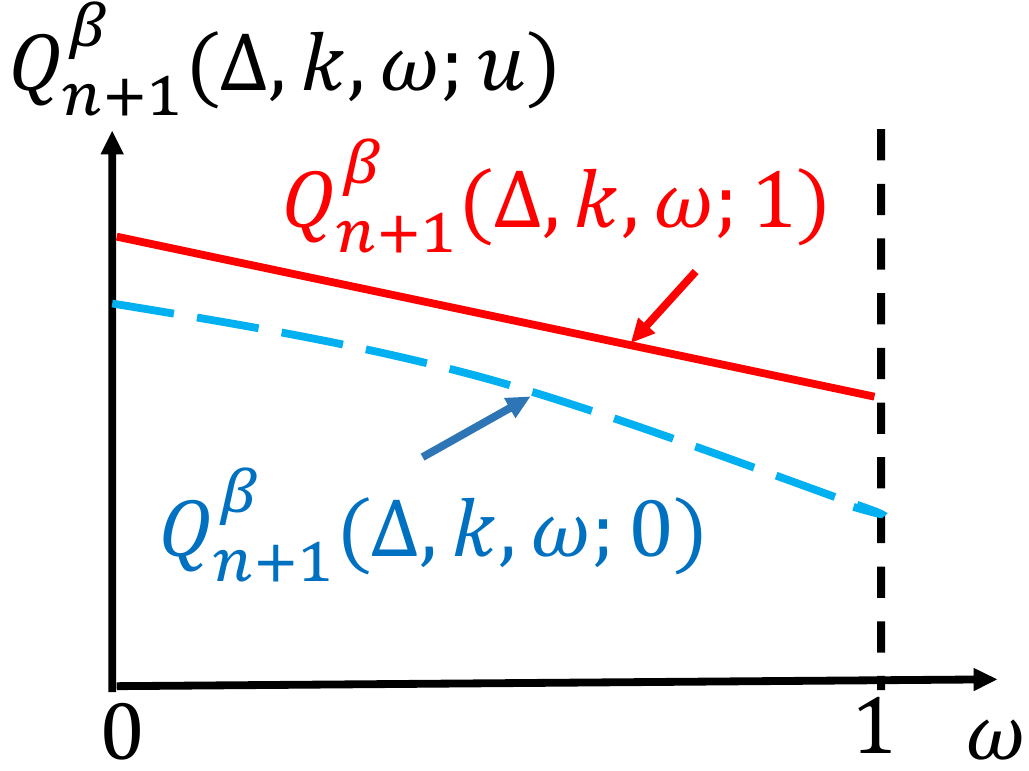}
 \label{case2}
 }
\caption{Values of $Q_{n+1}^\beta(\Delta,k,\omega;u)$}
\label{CompQ1Q2}
\end{figure}

\emph{Case 1: $Q_{n+1}^\beta(\Delta,k,1;1)\!\!<\! \!Q_{n+1}^\beta(\Delta,k,1;0)$ as in Fig. \ref{case1}.} Due to the concavity of $Q_{n+1}^\beta(\Delta,k,\omega;0)$ and linearity of $Q_{n+1}^\beta(\Delta,k,\omega;1)$ in $\omega$, there must be one unique intersection (corresponds to threshold). 

\emph{Case 2: $Q_{n+1}^\beta(\Delta,k,1;1)\!\geq\!\!Q_{n+1}^\beta(\Delta,k,1; 0)$ (see Fig. \ref{case2}):}. In the case, we will show that it is always optimal to suspend for any $\omega$ given $\Delta$ and $k$, i.e. $Q_{n+1}^\beta(\Delta,k,\omega;1)\geq Q_{n+1}^\beta(\Delta,k,\omega;0)$ for every $\omega$. In particular, by $Q_{n+1}^\beta(\Delta,k,1;1)\!\geq\!\!Q_{n+1}^\beta(\Delta,k,1; 0)$, and definitions \eqref{q1} and \eqref{q2}, we have $\lambda+\beta V_{n}^\beta(k, \left(k\right)_{+},p_{11})-\beta V_{n}^\beta(\Delta+1, \left(k\right)_{+},p_{11})\geq 0$. Moreover, by induction hypothesis, (c) holds for $n$. Thus, we have
\begin{align}
	&Q_{n+1}^\beta(\Delta,k,\omega;1)- Q_{n+1}^\beta(\Delta,k,\omega;0)\notag\\
		=&\omega\Big(\lambda+\beta V_{n}^\beta(k, \left(k\right)_{+},p_{11})-\beta V_{n}^\beta(\Delta+1, \left(k\right)_{+},p_{11})\Big)\notag\\
		&
		+\beta\Big(\omega V_n^\beta(\Delta+1, \left(k\right)_{+}, p_{11})-V_n^\beta(\Delta+1,\left(k\right)_{+}, \mathcal{T}(\omega))\notag\\
		&(1-\omega)V_n^\beta(\Delta+1, \left(k\right)_{+}, p_{01})\Big) +(1-\omega)\lambda \label{i1_1}\\
		\geq & 0\label{i1_2}
\end{align}

\emph{Step (2):} We show that (a)-(c) hold for $n\!+\!1$. 
First, we consider property (a). It suffices to show that if $\Delta'>\Delta$, then $V_{n+1}^\beta(\Delta',k,\omega)\geq V_{n+1}^\beta(\Delta,k,\omega)$. Since $V_{n+1}^\beta(\mathbf{s})=\min \{Q_{n+1}^\beta(\mathbf{s};1),Q_{n+1}^\beta(\mathbf{s};0)\}$, we only need to show that for any $u$ that applies to state $(\Delta',k,\omega)$, there exists an action $u'$ such that $Q_{n+1}^\beta(\Delta',k,\omega;u)\geq Q_{n+1}^\beta(\Delta,k,\omega;u')$

If $u=0$, then we have 
	\begin{align}
		&Q_{n+1}^\beta(\Delta', k, \omega;0)\notag\\
		=&\Delta'+\beta V_n^\beta(\Delta'+1,(k)_{+},\mathcal{T}(\omega))\label{i1_3}\\
		\geq &  \Delta+\beta V_n^\beta(\Delta+1,(k)_{+},\mathcal{T}(\omega))\label{i1_4}\\
		= & Q_{n+1}^\beta(\Delta, k, \omega;0)\label{i1_5}
	\end{align}
	The inequality \eqref{i1_4} holds since property (a) holds for $n$ by induction hypothesis.
	
If $u=1$, according to values of $\Delta$, we have two cases to consider specified as follows.	
	If $\Delta<K$, then $\Delta=k-1$ and it implies that the receiver has received the latest status update generated at the beginning of the frame. In the case, the action chosen for state $(\Delta, k, \omega)$ is to suspend. Recall that $\Delta'=mK+k-1$ at the $k$-th slot of certain frame, where $m> 0$. For the case, we have
	\begin{align}
 &Q_{n+1}^\beta(\Delta+mK, k, \omega;1)\notag\\
		=&\Delta+mK+\lambda+\beta\Big( \omega V_n^\beta(k,(k)_{+},p_{11})\notag\\
		&+(1-\omega) V_n^\beta(\Delta+K+1,(k)_{+},p_{01})\Big)\\
		\geq & \Delta+\lambda
		+\beta\Big( \omega V_n^\beta(k,(k)_{+},p_{11})\notag\\
		&+(1-\omega) V_n^\beta(k,(k)_{+},p_{01})\Big)\label{i1_6}\\
		\geq & \Delta+\beta V_n^\beta(k,(k)_{+},\mathcal{T}(\omega)) \label{i1_7}\\
		=&Q_{n+1}^\beta(\Delta, k, \omega;0)
	\end{align}		
		The inequality \eqref{i1_6} holds since property (a) holds for $n$ by induction hypothesis. The inequality \eqref{i1_7} holds since property (c) holds for $n$ by induction hypothesis.
		
	If $\Delta\geq K$, then we have
		\begin{align}
		&Q_{n+1}^\beta(\Delta', k, \omega;1)\notag\\
		=&\Delta'+\lambda
		+\beta\Big( \omega V_n^\beta(k,\left(k\right)_{+},p_{11})\notag\\
		&+(1-\omega) V_n^\beta(\Delta'+1,\left(k\right)_{+},p_{01})\Big)\\
		\geq &\Delta+\lambda+\beta\Big( \omega V_n^\beta(k,\left(k\right)_{+},p_{11})\notag\\
		&+(1-\omega) V_n^\beta(\Delta+1,\left(k\right)_{+},p_{01})\Big) \label{i1_8}\\
		=&Q_{n+1}^\beta(\Delta, k, \omega;1)
	\end{align}	
	The inequality \eqref{i1_8} holds since property (a) holds for $n$ by induction hypothesis.

Second, we consider property (b). It suffices to show that if $\omega'\leq \omega$, then $V_{n+1}^\beta(\Delta,t,\omega')\geq V_{n+1}^\beta(\Delta,t,\omega)$ given $V_{n}^\beta$ has properties (a)-(c). The general idea to show this is same to that in proving property (a). 

Since $p_{11}\geq p_{01}$, $\mathcal{T}(\omega)=(p_{11}-p_{01})\omega+p_{01}$ is non-decreasing in $\omega$ and $\mathcal{T}(\omega')\leq\mathcal{T}(\omega)$. Then, we have 
	\begin{align}
		& Q_{n+1}^\beta(\Delta, k, \omega';0)\notag\\
		=&\Delta+\beta V_n^\beta(\Delta+1,\left(k\right)_{+},\mathcal{T}(\omega'))\\
		\geq & \Delta+\beta V_n^\beta(\Delta+1,\left(k\right)_{+},\mathcal{T}(\omega)) \label{i1_9}\\
		=&Q_{n+1}^\beta(\Delta, k, \omega;0)
	\end{align}
	The inequality \eqref{i1_9} holds since property (b) holds for $n$ by induction hypothesis.
	
Recall that for the $\left(k\right)_+$-th slot of certain frame, the smallest age is $k$. Then, $V_n^\beta(k, \left(k\right)_{+},p_{11})- V_n^\beta(\Delta+1,\left(k\right)_{+},p_{01})\leq 0$ since properties (a) and (b) in Lemma \ref{property} hold for $n$ by induction hypothesis. Hence, we have
    \begin{align}
		&Q_{n+1}^\beta(\Delta, k, \omega';u=1)\notag\\
		=&\Delta+\lambda+\beta\Big(V_n^\beta(\Delta+1,\left(k\right)_{+},p_{01})\notag\\ 
		&+\omega' \left(V_n^\beta\left(k,\left(k\right)_{+},p_{11}\right)\!-\!V_n^\beta\left(\Delta+1,\left(k\right)_{+},p_{01}\right)\right)\Big)\\
		\geq &\Delta+\lambda+\beta\Big(V_n^\beta(\Delta+1,\left(k\right)_{+},p_{01})\notag\\
		&+\omega (V_n^\beta(k, \left(k\right)_{+},p_{11})-V_n^\beta(\Delta+1,\left(k\right)_{+},p_{01})) \Big)\label{i1_10}\\
		=&Q_{n+1}^\beta(\Delta, k, \omega;u=1)
	\end{align}	
The inequality \eqref{i1_10} holds since $\omega'\leq\omega$ and $V_n^\beta(k, \left(k\right)_{+},p_{11})- V_n^\beta(\Delta+1,\left(k\right)_{+},p_{01})\leq 0$.

Finally, we consider property (c). Note that $x\geq y$ and $z=\omega x+(1-\omega)y$.
For the left-hand-side of Eq. \eqref{prop 2}, there are three possible combinations of actions for state $(\Delta, k, x)$ and $(\Delta, k, y)$, i.e. suspending for both states, transmitting for both states and suspending for latter state but transmitting for former state. Note that $x\geq y$ implies that if the optimal action for state $(\Delta, k, y)$ is to update, then the optimal action for state $(\Delta, k, x)$ is also to update since the optimal policy for $n+1$-th iteration is of threshold type.

For the case of suspending for both states, we have 
\begin{align}
	& (1-\omega)\lambda+\omega Q_{n+1}^\beta(\Delta, k, x; 0)+ (1-\omega) Q_{n+1}^\beta(\Delta, k, y; 0)\notag\\
	=&(1-\omega)\lambda+\omega\left(\Delta+\beta V_n^\beta \left(\Delta+1,\left(k\right)_{+},\mathcal{T}(x)\right)\right)\notag\\
	&+(1-\omega)\left(\Delta+\beta V_n^\beta(\Delta+1,\left(k\right)_{+},\mathcal{T}(y))\right)\\
	\geq &  \Delta+\beta V_n^\beta(\Delta+1,\left(k\right)_{+},\mathcal{T}(z))\label{i1_11}\\
	= & Q_{n+1}^\beta(\Delta, k,z;0)\\
	\geq & V_{n+1}^\beta(\Delta, k,z) \label{i1_12}
\end{align}
The inequality \eqref{i1_11} holds since property (c) holds for $n$ by induction hypothesis. The inequality \eqref{i1_12} holds by \eqref{iteration}.

For the case of transmitting for both states, we have 
\begin{align}
	& (1-\omega)\lambda+\omega Q_{n+1}^\beta(\Delta, k, x; 1)+ (1-\omega) Q_{n+1}^\beta(\Delta, k, y; 1)\notag\\
	= & (1-\omega)\lambda+ \Delta+\lambda+\beta\Big( z V_n^\beta(k,\left(k\right)_{+},p_{11})\notag\\
	&+\!(1\!-\!z) V_n^\beta(\Delta+1,\left(k\right)_{+},p_{01})\Big)\\
	= & (1-\omega)\lambda+Q_{n+1}^\beta(\Delta, k,z;1)\\
	\geq &Q_{n+1}^\beta(\Delta, k,z;1) \label{i1_13}\\
	\geq & V_{n+1}^\beta(\Delta, k,z) \label{i1_14}
\end{align}
The first equality is by \eqref{q2} plus some basic calculation. The second equality is by \eqref{q2}. The inequality \eqref{i1_14} holds by \eqref{iteration}.

For the case of transmitting for state $(\Delta, k, x)$ but suspending for $(\Delta, k, y)$, we have 
\begin{align}
	& (1-\omega)\lambda+\omega Q_{n+1}^\beta(\Delta, k, x; 1)+ (1-\omega) Q_{n+1}^\beta(\Delta, k, y; 0)\notag\\
	=&\lambda+\Delta+\beta\omega\Big( x V_n^\beta(k, \left(k\right)_{+},p_{11})\notag\\
	&+(1-x) V_n^\beta(\Delta+1,\left(k\right)_{+},p_{01})\Big)\notag\\
	&+\beta (1-\omega)V_n^\beta(\Delta+1,\left(k\right)_{+},\mathcal{T}(y)) \label{i1_15} \\
	\geq & \lambda+\Delta+\beta\omega\Big( x V_n^\beta(k,\left(k\right)_{+},p_{11})\notag\\
	&+(1-x) V_n^\beta(\Delta+1,\left(k\right)_{+},p_{01})\Big)\notag\\
	&+\beta(1-\omega)\Big( y V_n^\beta(\Delta+1,\left(k\right)_{+},p_{11})\notag\\
	&+(1-y) V_n^\beta(\Delta+1,\left(k\right)_{+},p_{01})\Big) \label{i1_16}\\
	\geq & \lambda+ \Delta+\beta\Big( z V_n^\beta(k,\left(k\right)_{+},p_{11})\notag\\
	&+(1-z) V_n^\beta(\Delta+1,\left(k\right)_{+},p_{01})\Big)\label{i1_17}\\
	= &Q_{n+1}^\beta(\Delta, k,z;1)\\
	\geq & V_{n+1}^\beta(\Delta, k,z) \label{i1_110}
\end{align}
The equality \eqref{i1_15} is by \eqref{q1} and \eqref{q2}.
The inequality \eqref{i1_16} holds since the value function is a piecewise linear and concave function with respect to the belief state, which can be verified with theory developed in \cite{smallwood1973optimal}. The inequality \eqref{i1_17} holds since property (a) holds for $n$ by induction hypothesis with some basic calculation. The inequality \eqref{i1_110} holds by \eqref{iteration}.

\section{Proof for Verification of Conditions in \cite{sennott1989average}}
\label{exist_1}
	The conditions are listed below:
	\begin{itemize}
		\item A1: $V^\beta(\mathbf{s})$ defined in \eqref{disc_cost_opt} is finite $\forall \mathbf{s},\beta$.
		\item A2: $\exists L\geq 0$ s.t. $-L\leq h^\beta (\mathbf{s})\triangleq V^\beta(\mathbf{s})-V^\beta(\mathbf{0})$, $\forall \mathbf{s}, \beta$.
		\item A3: $\exists M(\mathbf{s})\geq 0$ s.t. $h^\beta (\mathbf{s})\leq M(\mathbf{s})$, $\forall \mathbf{s},\beta$. Moreover, for each $\mathbf{s}$, $\exists \, u(\mathbf{s})$ s.t. $\sum_{\mathbf{s'}\in \mathcal{S}}\mathbb{P}(\mathbf{s'}|\mathbf{s},u(\mathbf{s}))M(\mathbf{s'})<\infty$.
		\item A4: $\sum_{\mathbf{s'}\in \mathcal{S}}\mathbb{P}(\mathbf{s'}|\mathbf{s},u)M(\mathbf{s'})<\infty$ $\forall \mathbf{s}, u$.
	\end{itemize}
	In Proposition \ref{existence_discount}, we showed that a policy $f$ that chooses $u=0$ at every time slot satisfies $L_{\mathbf{s}}^\beta(f;\lambda)\!<\!\!\infty$. By \eqref{disc_cost_opt}, we have $V^\beta(\mathbf{s})\leq L_{\mathbf{s}}^\beta(f;\lambda)$, which implies A1. Moreover, we have $V^\beta$ increasing in $\Delta$ and decreasing in $\omega$ by Lemma \ref{property}. Hence, by setting $L=V^\beta(\mathbf{0})\!-\!\min_{k\in \mathcal{K}}V^\beta((k)_{-},k,p_{11})\!\geq \!0$, where $\mathbf{0}=(K,1,p_{11})$ is the reference state, we proves A2.

  Let $\delta$ be the policy that transmits at each time slot. Similar to proof of Lemma 6 in \cite{hsu2017scheduling}, The AoI can be regarded as a stable AoI queue. In particular, average arrival rate is one since age increases by 1 at each time slot, and average service rate is infinite since the channel is in a good state with positive probability and can serve infinite number of age packets when it is in a good state. In the case, the age queue is stable. Hence, states that occur after delivery are recurrent. This implies that $\mathbf{0}$ is recurrent. Actually, the probability of not entering state $\mathbf{0}$ after $l$ frames is no more than $b^{lK}$, where $b$ is steady state probability that channel is in a bad state. 
 Hence, under policy $\delta$ the expected cost of the first passage from state $\mathbf{s}$ to $\mathbf{0}$, denoted by $c_{\mathbf{s},\mathbf{0}}(\delta)$, is finite.
 Let $\delta'$ be a mix policy where $\delta$ is used until entering state $\mathbf{0}$ and the discounted Lagrange cost optimal policy $\delta_{\beta}$ is used afterwards. Suppose $T$ is the first time slot when system enters $\mathbf{0}$. Then, we have
   \begin{align}
   	&V^\beta(\mathbf{s})\notag\\
   	\leq &\mathbb{E}_{\delta'}[\sum_{t=1}^{T-1}\beta^{t-1}C(\mathbf{s}_t,\!u_t)|\mathbf{s}]\!+\!\mathbb{E}_{\delta'}[\sum_{t=T}^{\infty}\beta^{t-1}C(\mathbf{s}_t,\!u_t)|\mathbf{0}]\\
   	\leq & c_{\mathbf{s},\mathbf{0}}(\delta)+ \mathbb{E}_{\delta_{\beta}}(\beta^{(T-1)})V^\beta(\mathbf{0})\\
   	\leq & c_{\mathbf{s},\mathbf{0}}(\delta)+ V^\beta(\mathbf{0}).
   \end{align}   
   Hence, by setting $M(\mathbf{0})\!\!=\!0$ and $M(\mathbf{s})\!=\!\!c_{\mathbf{s},\mathbf{0}}(\delta)$ for $\mathbf{s}\!\neq \!\mathbf{0}$, we proves A3. After transition from $\mathbf{s}$ under any action, there will be at most two possible states. Since for all $\mathbf{s}$, $M(\mathbf{s})<\infty$, the sum of at most two $M(\cdot)$ is also finite. Hence, A4 holds.

\section{Proof of Theorem \ref{approx}}
   \label{convg}
	Let $V^{\beta,N}$ be the minimum $\beta$-discounted Lagrangian cost for the approximate MDP with bound $N$ and $h^{\beta,N}(\mathbf{s})=V^{\beta,N}(\mathbf{s})-V^{\beta,N}(\mathbf{0})$. By \cite{sennott1997computing}, it suffices to verify the following conditions B1-B2.
	\begin{itemize}
		\item B1: $\exists$ $L\geq 0$, $M(\cdot)\geq 0$ on $\mathcal{S}$ s.t. $-L\leq h^{\beta,N}(\mathbf{s})\leq M(\mathbf{s})$ for $\mathbf{s} \in \mathcal{S}^{N}$, where $\beta \in (0,1)$ and $N=K+1,K+2,\cdots$.
		\item B2: $\limsup_{N \rightarrow \infty} \bar{L}^{N*}(\lambda)\leq \bar{L}^*(\lambda)$.
	\end{itemize}

Consider policy $\pi$ that updates at each time slot with equal probability. Let $c_{\mathbf{s},\mathbf{0}}(\pi)$ and $c_{\mathbf{s},\mathbf{0}}^{N}(\pi)$ be the expected cost of the first passage from state $\mathbf{s}$ to $\mathbf{0}$ by applying $\pi$ to original and approximate MDP, respectively. 
Similar to the proof in Appendix \ref{exist_1}, we have $L=V^{\beta,N}(\mathbf{0})-\min_{k\in \mathcal{K}}V^{\beta,N}((k)_{-},k,p_{11})$, $c_{\mathbf{s},\mathbf{0}}(\pi)<\infty$ and $h^{\beta,N}(\mathbf{s})\leq c_{\mathbf{s},\mathbf{0}}^{N}(\pi)$. Next, we show that $c_{\mathbf{s},\mathbf{0}}^{N}(\pi)\leq c_{\mathbf{s},\mathbf{0}}(\pi)$. Then, $M(\mathbf{s})=c_{\mathbf{s},\mathbf{0}}(\pi)$. By the proof of Corollary 4.3 in \cite{sennott1997computing}, 
it suffices to show that 
\begin{align}
	\sum_{\mathbf{s}'\in \mathcal{S}^N}P^N_{\mathbf{s}\mathbf{s}'}(u)c_{\mathbf{s}',\mathbf{0}}(\pi)\leq \sum_{\mathbf{s}'\in \mathcal{S}}P_{\mathbf{s}\mathbf{s}'}(u)c_{\mathbf{s}',\mathbf{0}}(\pi)\label{i3_0}
\end{align}
Recall that $\nu$ is approximation operation to the state defined in \eqref{approximation_opr}. Then, we have
	\begin{align}
		&\sum_{\mathbf{s}'\in \mathcal{S}^N}\!\!P^N_{\mathbf{s}\mathbf{s}'}(u)c_{\mathbf{s}',\mathbf{0}}(\pi)
		\!\notag\\
		=&\sum_{\mathbf{s}'\in \mathcal{S}^N}\!\!\Big(P_{\mathbf{s}\mathbf{s}'}(u)\!+\!\!\!\sum_{\mathbf{r}\in \mathcal{S}-\mathcal{S}^{N}}P_{\mathbf{s}\mathbf{r}}(u)\mathbbm{1}_{\{\nu(\mathbf{r})=\mathbf{s}'\}}\Big)c_{\mathbf{s}',\mathbf{0}}(\pi)\\
		\leq & \sum_{\mathbf{s}'\in \mathcal{S}^N}\!\!P_{\mathbf{s}\mathbf{s}'}(u)c_{\mathbf{s}',\mathbf{0}}(\pi)\!+\!\!\!\!\!\!\!\sum_{\mathbf{r}\in \mathcal{S}-\mathcal{S}^N}\!\!\!\!\!P_{\mathbf{s}\mathbf{r}}(u)c_{\mathbf{r},\mathbf{0}}(\pi)\label{i3_1}\\
		=&\sum_{\mathbf{s}'\in \mathcal{S}}P_{\mathbf{s}\mathbf{s}'}(u)c_{\mathbf{s}',\mathbf{0}}(\pi)
	\end{align}
The inequality \eqref{i3_1} holds since policy $\pi$ does not depend on states  
 and thus $c_{(\Delta, k,\omega),\mathbf{0}}(\pi)\leq c_{(\Delta', k, \omega'),\mathbf{0}}(\pi)$ for $\Delta\leq \Delta'$.	

For B2, 
	claim that $V^{\beta,N}(\mathbf{s})\leq V^\beta(\mathbf{s})$ for all $N$. Then, for all $N$.
		$\bar{L}^{N*}(\lambda)\!\!=\!\!\lim_{\beta\rightarrow 1}(1\!-\!\beta)V^{\beta,N}(\mathbf{s})
		\leq \lim_{\beta\rightarrow 1}(1\!-\!\beta)V^\beta(\mathbf{s})\!=\!\bar{L}^*(\lambda).$
	We use induction to prove the claim. The claim holds obviously when $n=0$. Suppose $V_n^{\beta,N}(\mathbf{s})\leq V_n^\beta(\mathbf{s})$, then 	
	\begin{align}
		&V_{n+1}^{\beta,N}(\mathbf{s})\notag\\
		=&\min_{u}\{C(\mathbf{s},u;\lambda)+\beta\sum_{\mathbf{s}'\in\mathcal{S}^{N}} P^N_{\mathbf{s}\mathbf{s}'}(u)V_n^{\beta,N}(\mathbf{s}')\}\\
		\leq &\min_{u}\{C(\mathbf{s},u;\lambda)+\beta\sum_{\mathbf{s}'\in\mathcal{S}^{N}} P^N_{\mathbf{s}\mathbf{s}'}(u)V_n^{\beta}(\mathbf{s}')\}\label{i3_2}\\
		\leq & \min_{u}\{C(\mathbf{s},u;\lambda)+\beta \sum_{\mathbf{s}'\in \mathcal{S}}P_{\mathbf{s}\mathbf{s}'}(u)V_n^{\beta}(\mathbf{s}')\}\label{i3_3}\\
		=&V_{n+1}^{\beta}(\mathbf{s})		
	\end{align}
The inequality \eqref{i3_2} is due to the induction hypothesis. The inequality \eqref{i3_3} is due to $\sum_{\mathbf{s}'\in \mathcal{S}^N}P^N_{\mathbf{s}\mathbf{s}'}(u)V_n^{\beta}(\mathbf{s}')\leq \sum_{\mathbf{s}'\in \mathcal{S}}P_{\mathbf{s}\mathbf{s}'}(u)V_n^{\beta}(\mathbf{s}')$, which can be shown similar to \eqref{i3_0}.

\section{Proof of Lemma \ref{property2}}
\label{mono_property_case}
 Let $V_n^\beta(\mathbf{s})$ be the cost-to-go function such that $V_0^\beta(\mathbf{s})=0$ for all $\mathbf{s}\in\mathcal{S}$ and for $n\geq 0$,
		\begin{align}
		V_{n+1}^\beta(\Delta,k,g)=\min_{u\in \{0,1\}} Q_{n+1}^\beta(\Delta,k,g;u)
	    \end{align}
	    where
	    \begin{align}
		\!\!Q_{n+1}^\beta\left(\Delta,k,g;0\right)=&\Delta\!+\!\beta\!\!\!\!\! \sum_{g'\in\{0,1\}}\!\!\!p_{gg'}V_{n}^\beta\left(\Delta\!+\!1,\left(k\right)_{+},g'\right)\label{u0}\\
		\!\!Q_{n+1}^\beta\left(\Delta,k,g;1\right)=&\Delta+\lambda+\beta \Big(p_{g1} V_{n}^\beta\left(k, \left(k\right)_{+},1\right)\notag\\&+p_{g0}V_{n}^\beta
		\left(\Delta+1,\left(k\right)_{+},0\right)\Big)\label{u1}
	    \end{align}
	     With similar argument in proof of Proposition \ref{existence_discount}, we can obtain that $V_n^\beta(\mathbf{s}) \rightarrow V^\beta(\mathbf{s})$ as $n\rightarrow \infty$, for every $\mathbf{s}$, $\beta$. Hence, we only need to show that for all $n$, the function $V_n^\beta(\Delta,k,g)$ is non-decreasing in AoI. Next, we show the result using induction. Note that zero function (i.e., $V_0^\beta(\mathbf{s})=0$) satisfies the property. In other words, for $n=0$, the property holds. Suppose that the property holds for $n$. It remains to show that the property holds for $n+1$.
Suppose $\Delta'>\Delta$, we will show $V_{n+1}^\beta(\Delta',k,g)\geq V_{n+1}^\beta(\Delta,k,g)$. Since $V_{n+1}^\beta(\mathbf{s})=\min \{Q_{n+1}^\beta(\mathbf{s};1),Q_{n+1}^\beta(\mathbf{s};0)\}$, it suffices to show that for each $u$ that applies to state $(\Delta',k,g)$, there exists an action $u'$ such that $Q_{n+1}^\beta(\Delta',k,g;u)\geq Q_{n+1}^\beta(\Delta,k,g;u')$.

If $u=0$, then we have
	\begin{align}
		&Q_{n+1}^\beta(\Delta', k, g;0)\notag\\
		=&\Delta'+\beta \sum_{g'\in\{0,1\}}p_{gg'} V_n^\beta(\Delta'+1,(k)_{+},g')\\
		\geq & \Delta+\beta \sum_{g'\in\{0,1\}}  p_{gg'}V_n^\beta(\Delta+1,(k)_{+},g')\label{i4_1}\\
		=& Q_{n+1}^\beta(\Delta, k, g;0)
	\end{align}
The inequality \eqref{i4_1} holds by our induction hypothesis.

If $u=1$, then we have two cases to consider based on the values of $\Delta$.
	At the $k$-th slot of a time frame, if $\Delta<K$, then $\Delta=k-1$ and it implies that the receiver has received the latest status update generated at the beginning of the frame. In the case, the action is to suspend. Recall that $\Delta'=mK+k-1$ at the $k$-th slot of certain frame, where $m> 0$. For the case, we have
	\begin{align}
 &Q_{n+1}^\beta(mK+k-1, k, g;1)\notag\\
		=&mK+k-1+\lambda+\beta\Big( p_{g1} V_n^\beta(k,(k)_{+},1)\notag\\
		&+p_{g0} V_n^\beta(mK+k,(k)_{+},0)\Big)\\
		\geq & k-1+\beta\left( p_{g1} V_n^\beta(k,(k)_{+},1)+p_{g0} V_n^\beta(k,(k)_{+},0)\right) \label{i4_2}\\
		=&Q_{n+1}^\beta(k-1, k, g;0)\\
		=&Q_{n+1}^\beta(\Delta, k, g;0)
	\end{align}	
	The inequality \eqref{i4_2} holds by induction hypothesis.
	
	If $\Delta\geq K$, then we have
		\begin{align}
		&Q_{n+1}^\beta(\Delta', k, g;1)\notag\\
		=&\Delta'+\lambda
		+\beta\Big( p_{g1} V_n^\beta(k,(k)_{+},1)\notag\\
		&+p_{g0} V_n^\beta(\Delta'+1,(k)_{+},0)\Big)\\
		\geq &\Delta+\lambda+\beta\Big(p_{g1} V_n^\beta(k,\left(k\right)_{+},1)\label{i4_3}\notag\\
		&+p_{g0} V_n^\beta(\Delta+1,\left(k\right)_{+},0)\Big)\\
		=&Q_{n+1}^\beta(\Delta, k, g;1)
	\end{align}	
	The inequality \eqref{i4_3} holds by induction hypothesis.
	
\section{Proof of Lemma \ref{threshold_prop2}}	
\label{threshold_discount2}

	Without loss of generality, we assume that at state $(\Delta,k,g)$ it is optimal to attempt a transmit. That is, $Q^\beta(\Delta,k,g;1)\leq Q^\beta(\Delta,k,g;0)$. Then, for any $\Delta'>\Delta$, 
	\begin{align}
		&Q^\beta(\Delta',k, g;1)-Q^\beta(\Delta',k,g;0)\notag\\
		=&\lambda+\beta p_{g1}(V^\beta\left(k, \left(k\right)_{+},1\right)-V^\beta\left(\Delta'+1, \left(k\right)_{+},1\right))\\
		\leq &\lambda+\beta p_{g1}(V^\beta\left(k, \left(k\right)_{+},1\right)-V^\beta\left(\Delta+1, \left(k\right)_{+},1\right))\label{i5_1}\\
		=&Q^\beta(\Delta,k,g;1)- Q^\beta(\Delta,k,g;0)\\
		\leq &0
	\end{align}
	The inequality \eqref{i5_1} holds since $V^\beta\left(\Delta'+1, \left(k\right)_{+},1\right)\geq V^\beta\left(\Delta+1, \left(k\right)_{+},1\right)$ by Lemma \ref{property2}. 
	Thus, it is also optimal to transmit at $(\Delta',k,g)$. Hence, the unconstrained discounted Lagrange cost optimal policy is of threshold-type in AoI.
	
	Let $\Delta_{\beta}^*(k,g;\lambda)$ denote the threshold associated with $k$ and $g$. That is, given $k$ and $g$, it is optimal to transmit when $\Delta\geq\Delta_{\beta}^*(k,g;\lambda)$. Let $\Delta_1=\Delta_{\beta}^*(k,0;\lambda)$, we have
	\begin{align}
		&Q^\beta(\Delta_1,k, 1;1)-Q^\beta(\Delta_1,k,1;0)\notag\\
		=&\lambda+\beta p_{11}\big(V^\beta\left(k, \left(k\right)_{+},1\right)-V^\beta\left(\Delta_1+1, \left(k\right)_{+},1\right)\big)\\
		\leq &\lambda+\beta p_{01}\big(V^\beta\left(k, \left(k\right)_{+},1\right)-V^\beta\left(\Delta_1+1, \left(k\right)_{+},1\right)\big)\label{i5_2}\\
		=&Q^\beta(\Delta_1,k,0;1)- Q^\beta(\Delta_1,k,0;0)\\
		\leq &0. \label{i5_3}
	\end{align}
	The inequality \eqref{i5_2} holds since $p_{01}\leq p_{11}$ by assumption and $V^\beta\left(k, \left(k\right)_{+},1\right)-V^\beta\left(\Delta+1, \left(k\right)_{+},1\right)\leq 0$ by Lemma \ref{property2}. The inequality \eqref{i5_3} holds by optimality. 
	
	Thus, we have $\Delta_{\beta}^*(k,0;\lambda)=\Delta_1\geq\Delta_{\beta}^*(k,1;\lambda)$. In other words, the threshold associated with good state is not larger than that associated with bad state.

\bibliographystyle{unsrt}
\bibliography{References}
\end{document}